\documentclass[10pt]{article} %%{proc}
\usepackage{latexsym,amsmath}
\usepackage{lscape}
\usepackage{epstopdf,caption,subcaption,graphicx,xcolor}
\usepackage{tikz}
\usepackage{pgfgantt}
\usepackage{hyperref}
\usepackage[capitalise,noabbrev]{cleveref}
\usetikzlibrary{shapes,arrows,fit,calc,positioning}
\usetikzlibrary{decorations.pathreplacing}
\normalsize
\newtheorem{theorem}{Theorem}
\newtheorem{lemma}{Lemma}

\newtheorem{observation}{Observation}

\newtheorem{definition}{Definition}

\newenvironment{proof}{{\sc Proof. }}{\hfill$\Box$\vspace{0.1in}}
\def\mcH{\mathcal{H}}
\def\mcP{\mathcal{P}}
\def\mcQ{\mathcal{Q}}
\def\mcS{\mathcal{S}}
\title{Covering vertices by sequential stars}
\author{
Mengyuan~Hu\thanks{Department of Mathematics, Hangzhou Dianzi University. Hangzhou, China.
	Emails: \texttt{\{2341070105, anzhang, chenyong, 221070019\} @hdu.edu.cn}} $^{\P}$
\and
An~Zhang$^{*}$\thanks{Corresponding authors.}
	\thanks{ORCID: \texttt{0000-0002-2622-5158}}
\and
Yong~Chen$^*$\thanks{ORCID: \texttt{0000-0001-8982-7757}}
\and
Zhikai~Chen\thanks{Department of Computing Science, University of Alberta.  Edmonton, Canada.
	Emails: \texttt{\{mengyua6, zhikai2, guohui, jm29, ys29\} @ualberta.ca}}
\and
Wei~Ding\thanks{School of Computer Science and Technology, Zhejiang University of Water Resources and Electric Power. Hangzhou, China.
	Email: \texttt{dingweicumt@163.com}; ORCID: \texttt{0000-0002-4353-5385}}
\and
Guohui~Lin$^{\dagger \P}$\thanks{ORCID: \texttt{0000-0003-4283-3396}}
\and
Jiaxuan~Ma$^{* \P}$
\and
Yue~Sun$^\P$
	\thanks{Institute of Operations Research and Information Engineering, Beijing University of Technology. Beijing, China.
	Email: \texttt{sunyue@emails.bjut.edu.cn}}
}
\date{}
\begin{document}
\maketitle
%==================================================================================================
\begin{abstract}
We study the problem of covering the maximum number of vertices in a graph by a collection of vertex-disjoint stars, each with a number of satellites in a given interval $[k, \ell]$,
where $1 \le k < \ell$ and $\ell$ can be infinity.
This is referred to as sequential {\sc $[k, \ell]$-Star Packing} problem.
It is solvable in polynomial time when $k = 1$, but becomes strongly NP-hard when $k \ge 2$.
In this paper, we propose either the first or an improved approximation algorithm for the following four sequential settings:
\begin{itemize}
\parskip=0pt
\item
	a $\frac {k+1}2$-approximation algorithm when $k \ge 3$ and $\ell = \infty$, improving the previous best ratio of $\frac {(k+1)^2}{2k+1}$; %-approximation algorithm;
\item
	a $\frac 43$-approximation algorithm when $k = 2$ and $\ell = \infty$, improving the previous best ratio of $\frac 32$; %-approximation algorithm;
\item
	the first $(1 + \frac \ell{\ell+1})$-approximation algorithm when $2 = k < \ell$;
\item
	and the first $(1 + \max\left\{\frac {k-1}2, \frac {(k+1) \ell}{3 (\ell+1)}\right\})$-approximation algorithm when $3 \le k < \ell$.
\end{itemize}
Besides the main algorithmic techniques being local search coupled with amortized analysis,
we observe augmenting configurations to bridge two distant neighborhoods for a local improvement operation.
Additionally, the problem has been shown APX-hard when $k \ge 3$;
we prove its APX-hardness for the last remaining case where $k = 2$.

\end{abstract}
%==================================================================================================

\section*{Acknowledgments}
%--------------------------------------------------------------------------------------------------
The research is supported by 2026XMHD034,
the NNSF of China (Grants 12371316, 12471301, 92567201),
the PNSF of Zhejiang China (Grant LZ25A010001),
the China Scholarship Council (Grant Nos. 202306540092, 202508330173, 202508330175),
the Nanxun Scholars Program of ZJWEU (No. RC2023010481), and
the NSERC Canada.

%\newpage
\section{Introduction}
%==================================================================================================
Given an integer $i \ge 1$, an {\em $i$-star} is a complete bipartite graph in which one part contains one vertex, called the {\em center} of the star,
and the other part contains $i$ vertices called the {\em satellites}.
An $i$-star is also called a $k^+$-star if $i \ge k$, or an $\ell^-$-star if $i \le \ell$.
The star packing problem is to find a collection of vertex-disjoint stars in a given graph $G = (V, E)$ so that the number of vertices in these stars is maximized.
One sees that, if any star can be used, then the problem is easy, i.e., solvable in polynomial time.
The more interesting case is where there are constraints on the candidate stars,
leading to many variants that have been studied in the literature for their theoretical importance and a wide range of applications such as wireless sensor networking~\cite{LCC17}.

We examine the {\em sequential} star packing problem in this paper,
here by sequential it means the candidate stars have their numbers of satellites in a given interval $[k, \ell]$, for $1 \le k \le \ell$ and $\ell$ can be a fixed integer or infinity.
The problem is denoted as {\sc $[k, \ell]$-Star Packing}.
We remark that in the literature, sequential was used to refer to the case where $k = 1$ only~\cite{HK86}.
Hell and Kirkpatrick showed both {\sc $[1, \infty)$-Star Packing} and {\sc $[1, \ell]$-Star Packing}, for any fixed integer $\ell$, to be polynomial time solvable~\cite{HK86},
and the star packing problem using any other candidate star family to be strongly NP-hard~\cite{HK84}.
It follows that {\sc $[k, \ell]$-Star Packing} is hard for any $k \ge 2$, and we study it in this paper from the approximation algorithm perspective.

When $k = \ell$, that is, only one candidate star is allowed, {\sc $[k, k]$-Star Packing} reduces to the $(k+1)$-Set Packing problem
(by constructing a $(k+1)$-set of vertices whenever these $k+1$ vertices can form into a $k$-star, in $O(|V|^{k+1})$ time),
the latter of which admits a $(\frac {k+2}3 + \epsilon)$-approximation algorithm, for any $\epsilon > 0$~\cite{FY14}.
It turns out the design and analysis techniques we develop in this paper do not work out a better approximation algorithm, and thus we assume $k < \ell$ throughout the paper,
i.e., in each variant we have at least two candidate stars.
We also observe the case where $k = 2$ to be quite different from the other case where $k \ge 3$.
Therefore, we distinguish four distinct variants, namely, {\sc $[2, \infty)$-Star Packing}, {\sc $[k, \infty)$-Star Packing} where $k \ge 3$,
{\sc $[2, \ell]$-Star Packing} where $\ell > 2$, and {\sc $[k, \ell]$-Star Packing} where $\ell > k \ge 3$.

The {\sc $[k, \infty)$-Star Packing} problem is also denoted as {\sc $k^+$-Star Packing}, and has received several studies.
It is NP-hard for any $k \ge 2$ even on bipartite graphs~\cite{LL21}.
Specially, {\sc $[2, \infty)$-Star Packing} remains NP-hard on even more restricted graph classes such as bipartite graphs of maximum degree $4$~\cite{LL21} and cubic graphs~\cite{XLL24}.
On the inapproximability, {\sc $[3, 3]$-Star Packing} on cubic graphs is observed APX-complete~\cite{XL24},
since it is equivalent to the maximum distance-$3$ independent set problem~\cite{EIL17}.
Recently, Li~\cite{Li25} proved a lower bound of $\frac {12k + 19}{12k + 18.994}$ on the approximation ratio for {\sc $[k, \infty)$-Star Packing} when $k \ge 3$, if P $\ne$ NP.

A number of approximation algorithms have been designed for {\sc $[k, \infty)$-Star Packing}.
Li and Lin~\cite{LL21} presented the first $2$-approximation algorithm for $k = 2$,
and the approximation ratio was later improved to $\frac 95$ by Huang et al.~\cite{HZG25} and to $\frac 32$ by Hu et al.~\cite{HZC25,HZC26} which is the state-of-the-art.
When $k \ge 3$, Huang et al.~\cite{HZG25} proposed a local search $(1 + \frac k2)$-approximation algorithm,
Li~\cite{Li25} proposed an almost the same algorithm independently,
and Hu et al.~\cite{HZC25,HZC26} designed an improved $(1 + \frac {k^2}{2k+1})$-approximation algorithm which is the state-of-the-art.
We note that there are approximation algorithms proposed for other variants and/or on special graph classes~\cite{LL21,XLL24,HZC25,HZC26}.

\subsection{Other related work}
%--------------------------------------------------------------------------------------------------
The star packing problem belongs to the general $\mcH$-packing framework, where $\mcH$ is a family of (non-isomorphic) graphs.
Given an $\mcH$, an $\mcH$-packing in an input graph $G = (V, E)$ is a set of vertex-disjoint subgraphs of $G$, each isomorphic to a graph in $\mcH$.
A vertex in $G$ is said to be {\em covered} by the $\mcH$-packing if it is contained in one of the subgraphs in the packing.
The $\mcH$-packing problem aims to find an $\mcH$-packing to cover the maximum number of vertices.

For another example of $\mcH$, we can look at an $i$-clique, typically denoted as $K_i$, which is the complete graph on $i$ vertices.
When $\mcH = \{K_2\}$, i.e., it contains only a single $2$-clique, the $\mcH$-packing problem is equivalent to the maximum matching problem.
A maximum matching can be computed by several well-known algorithms including the Blossom algorithm by Edmonds~\cite{Edm65b}.
The Blossom algorithm is iterative, that it starts with an empty matching and in each iteration it looks for an augmenting path (after contracting odd cycles),
which connects two uncovered vertices through an alternating path with edges alternately not in the current matching and in the current matching;
the algorithm terminates when no augmenting path is found, and the achieved matching is shown to be optimal.
The Blossom algorithm runs in $O(|V|^2 \cdot |E|)$ time, that is, the $\mcH$-packing problem is polynomial solvable.
When $\mcH = \{ K_2, K_3 \}$, 
Hell and Kirkpatrick~\cite{HK84} discovered three specific augmenting structures,
and proved that any packing, in which none of the three structures exists, is a packing that covers the maximum number of vertices.
They provided a polynomial time algorithm for computing such an optimal $\mcH$-packing.
Lastly, they showed if $\mcH \subseteq \{K_1, K_2, \ldots \}$ contains $K_1$ or $K_2$,
then the $\mcH$-packing problem is polynomial solvable,
and the provided algorithm is also based on augmenting structures connecting different uncovered vertices or some other specific local structures~\cite{HK84}.

In another paper, Hell and Kirkpatrick~\cite{HK86} did the same for {\sc $[1, \ell]$-Star Packing} where $\ell$ is a fixed integer,
by discovering certain augmenting structures on the $\ell$-stars in the current packing,
and proving that any packing in which none of structures exists is a packing that covers the maximum number of vertices;
they provided a polynomial time algorithm for computing such a packing.

%Similarly, a series of results also exist for the partition problem with bounded sequential star sets.
%Given an input graph $G=(V,E)$,
%the $k^-$-star partition problem is to find a minimum collection of vertex-disjoint stars, each of which has at most $k$ satellites, to cover all the vertices of $V$.
%It is easy to see $1^-$-star partition problem is equivalent to the maximum matching problem, which is polynomially solvable\cite{GK04}.
%Since an $i$-star is equivalent to an $(i+1)$-path when $i\le2$, 
%the $2^-$-star partition problem is as same as $3^-$-path partition problem.
%The current best result is a $\frac{21}{16}$-approximation algorithm given by Chen et al.~\cite{CGS19}.
%For each $k\ge3$, Xu et al. designed a local search algorithm for $k^-$-star partition problem,
%which achieves an approximation ratio of $\frac{k+1}2 - \frac{k-1}{(k+1)(k+2)}$ when $k$ is odd and $\frac{k+1}2 - \frac{k-1}{2(k+1)^2}$ when $k$ is even\cite{XY25}.
%This ratio was improved to $\frac{k+1}2 - \frac{k-1}{8k-6}$ by Gong et al.\cite{GL25}.

\subsection{Our contributions and paper organization}
%--------------------------------------------------------------------------------------------------
In this paper, we study four NP-hard variants of the star packing problem from the approximation algorithm perspective, namely,
{\sc $[2, \infty)$-Star Packing}, {\sc $[k, \infty)$-Star Packing} where $k \ge 3$,
{\sc $[2, \ell]$-Star Packing} where $\ell > 2$, and {\sc $[k, \ell]$-Star Packing} where $\ell > k \ge 3$.
We continue to develop local search techniques coupled with the amortized analysis,
and we explore augmenting configurations that help bridge two distinct local structures for applying an improvement operation,
some of which are inspired by the work of Edmonds~\cite{Edm65b} and Hell and Kirkpatrick~\cite{HK84,HK86}.

We contribute two approximation algorithms for {\sc $[k, \infty)$-Star Packing} when $k = 2$ and $k \ge 3$, respectively, improving the previous best ones.
Specifically, for {\sc $[k, \infty)$-Star Packing} where $k \ge 3$, we present in Section 2 a $\frac {k+1}2$-approximation algorithm,
improving the previous best ratio of $\frac {(k+1)^2}{2k+1}$ by Hu et al.~\cite{HZC25,HZC26}.
Being the first algorithm in the paper, the description and the performance analysis are presented in full details.
Mostly building on top of the first algorithm, in the second half of Section 2, we propose additionally an augmenting configuration for {\sc $[2, \infty)$-Star Packing}
to design a $\frac 43$-approximation algorithm, improving the previous best ratio of $\frac 32$.
The {\sc $[k, \ell]$-Star Packing} problem, where $2 \le k < \ell$, does not seem to have been studied in the literature.
When $k = 2$, in Section 3, we design slightly different local improvement operations from those in the first algorithm,
and propose a distinct augmenting configuration, leading to the first $(1 + \frac \ell{\ell+1})$-approximation algorithm.
Section 4 deals with $k \ge 3$, where we make full use of the augmenting configuration,
and show that the same algorithm is a $(1 + \max\left\{\frac {k-1}2, \frac {(k+1) \ell}{3 (\ell+1)}\right\})$-approximation algorithm.
Given that a lower bound of $\frac {12k + 19}{12k + 18.994}$ on the approximation ratio has been proven for {\sc $[k, \infty)$-Star Packing} when $k \ge 3$,
we prove in Section 5 the APX-hardness for the last remaining case of {\sc $[2, \infty)$-Star Packing}.
We conclude the paper in Section 6.

%\newpage
\section{Approximating the {\sc $[k, \infty)$-Star Packing} problem when $k \ge 2$}
%==================================================================================================
The {\sc $[k, \infty)$-Star Packing} problem is also referred to as $k^+$-star packing in the literature~\cite{LL21,HZG25,HZC25,HZC26}.
The state-of-the-art approximation algorithms are local search proposed by Hu et al.~\cite{HZC25,HZC26},
which have a ratio of $1 + \frac {k^2}{2k + 1}$ when $k \ge 3$ and a ratio of $\frac 32$ when $k = 2$.
In this section, we present improved, also local search, approximation algorithms {\sc Approx$[k, \infty)$} for $k \ge 3$ with a ratio of $\frac {k+1}2$,
and {\sc Approx$[2, \infty)$} for $k = 2$ with a ratio of $\frac 43$.
We will highlight the major differences after we introduce the new local improvement operations.

A local search algorithm is typically supported by a number of operations each is designed to improve the current solution through searching a neighborhood in the solution space.
Our algorithm is denoted as {\sc Approx$[k, \infty)$} and the operations are {\sc Collect} and {\sc Trade} defined below.
The algorithm is iterative, and during each iteration a feasible solution, referred to as the {\em current} solution or packing, is assumed at the beginning,
one of the operations is applied to update/improve the solution.
Before the iteration ends, post-processing is often necessary including collecting the solution information, updating the {\em remainder} graph,
and some other minor tasks to be detailed in each specific algorithm.
If in an iteration none of the operations is applicable, then the algorithm terminates and returns the current solution as the final solution.

Below we define the operations first, and for each operation we estimate its time complexity.
Given a feasible solution $\mcP$ for the input graph $G = (V, E)$, we call each star in $\mcP$ an {\em internal} star.
Let $V(\mcP)$ denote the set of vertices {\em covered} by (the stars of) $\mcP$,
then $V - V(\mcP)$ is the set of {\em uncovered} vertices and the graph $R(\mcP)$ induced on $V - V(\mcP)$, i.e., $R(\mcP) = G[V - V(\mcP)]$,
is the {\em remainder} graph with respect to $\mcP$.
A star in the remainder graph is called an {\em outside} star.

\begin{definition}[Operation {\sc Collect}]
\label{def01}{\rm \cite{HZG25,HZC25,HZC26}}
Given a vertex $v$ in the remainder graph, if its degree in the remainder graph is $d \ge k$,
then the operation extracts the $d$-star centered at $v$ with the $d$ neighbors of $v$ as satellites from the remainder graph and adds it to the current solution $\mcP$.
({\em Comment}: If $d < k$, then the operation is not applicable at the vertex $v$.)
\end{definition}

The algorithm starts with the empty solution,
and in the first a few iterations {\sc Collect} is repeatedly applied on vertices of degree at least $k$ in the remainder graph until impossible.
Note that the degrees of all the vertices can be pre-processed in $O(|E|)$ time,
and inside the algorithm their degrees in the remainder graph can be updated in $O(|E|)$ time too after applying a {\sc Collect} operation.
Each operation {\sc Collect} alone takes $O(|V|)$ time but the iteration takes $O(|E|)$ time due to post-processing.
When no {\sc Collect} operation is applicable, the maximum degree of the remainder graph is at most $k-1$, which is stated in the following Observation~\ref{obs01}.
We remark that an observation typically consists of straightforward statements that do not need rigorous proofs,
though we will provide some sample proofs below.

\begin{observation}
\label{obs01}
When no {\sc Collect} operation is applicable,
\begin{enumerate}
\parskip=0pt
\item the maximum degree of the remainder graph is at most $k-1$;
\item the center of an internal star is not adjacent to any uncovered vertex.
\end{enumerate}
\end{observation}

As we define more operations below, we remark that after applying any one of them in an iteration,
during the post-processing operation {\sc Collect} is repeatedly applied whenever there is an uncovered vertex of degree at least $k$ in the remainder graph,
or if there is an uncovered vertex adjacent to the center of an internal star then the uncovered vertex is added to the star to become covered.
Such post-processing at the end of each iteration guarantees the two statements in Observation~\ref{obs01}.

\begin{definition}[Operation {\sc Trade-$i$}]
\label{def02}
Given $i$ internal stars $S_1, \ldots, S_i$, where $1 \le i \le 3$, the operation trades them to cover more vertices by
first removing them from the current solution to make all their vertices uncovered and then extracting one up to $i+1$ feasible stars to cover more vertices.
({\em Comment}: Extracting such one up to $i+1$ feasible stars can be done by enumerating all subsets of $i+1$ uncovered vertices as the centers,
checking whether these centers have sufficiently many satellites and whether each has at least $k$ distinct satellites by a maximum bipartite matching algorithm,
in $O(|V|^{i+1 + 1.5})$ time.
If no combination covering more vertices can be extracted, then the operation is not applicable on this subset of internal stars $S_1, \ldots, S_i$.)
\end{definition}

To find an applicable operation {\sc Trade-$i$}, one might need to enumerate all subsets of $i$ internal stars.
Therefore, an applicable operation {\sc Trade-$i$} takes $O(|V|^{2i+2.5})$ time, for $1 \le i \le 3$.

\begin{observation}
\label{obs02}
When operation {\sc Trade-$1$} is not applicable, for every $1 \le j \le k$,
any $j$ satellites of an internal $(k+j)^+$-star together with any $k-j+1$ uncovered vertices do not form a $k$-star.
\end{observation}

More observations can be made when operation {\sc Trade-$i$} is not applicable, for $i = 2, 3$, and we will do it whenever we need them.
Using the above defined {\sc Collect} and {\sc Trade-$i$} operations, our algorithm {\sc Approx$[k, \infty)$} repeatedly applies them whenever possible to update the solution.
When none of them is applicable, the algorithm returns the achieved solution $\mcP$ as the final solution to the {\sc $[k, \infty)$-Star Packing} problem.

We remark that inside the previous best $(1 + \frac {k^2}{2k + 1})$-approximation algorithm~\cite{HZC26},
the same operation {\sc Collect} was deployed,
as well as some special cases of {\sc Trade-$1$} and {\sc Trade-$2$} where some structural constraints were put on the to-be traded internal stars and the newly extracted stars.
For example, for {\sc Trade-$1$}, it was required in \cite[Definitions 3 and 2 of {\sc Pull-by-$k$} and {\sc Pull-by-$(k+1)^+$}]{HZC26} that
either the internal star is a $k$-star, or exactly one of its satellites together with $k$ uncovered vertices form into a $k$-star.
Our operation {\sc Trade-$i$} is thus more general, imposing no specific structural constraints, while more powerful though also more time consuming.

\subsection{The performance analysis for $k \ge 3$}
%--------------------------------------------------------------------------------------------------
We fix an optimal star packing denoted as $\mcQ$ for discussion, in which the stars are referred to as {\em optimal} stars.
Recall that the stars in the computed star packing $\mcP$ are called {\em internal};
a vertex is said {\em covered} or {\em uncovered} with respect to $\mcP$, i.e., whether or not it is in an internal star.

Using Observation~\ref{obs01}(1), when the center of an optimal star is uncovered, at most $k-1$ of its satellites can be uncovered.
We categorize the optimal stars into three types as follows:
\begin{definition}[Types of optimal stars]
\label{def03}
The optimal stars are categorized into three types:
\begin{itemize}
\parskip=0pt
\item In a Type-1 optimal star, the center and $k - 1$ satellites are uncovered.
\item In a Type-2 optimal star, the center and no more than $k - 2$ satellites are uncovered.
\item In a Type-3 optimal star, the center is covered.
\end{itemize}
\end{definition}

Such a categorization of the optimal stars into three types is complete.

We associate one token with each uncovered vertex of $V(\mcQ)$, that is, the number of tokens in an optimal star equals to the number of uncovered vertices.
For each optimal star, the tokens are {\em transferred}, maybe fractionated,
to some of the covered vertices in the same optimal star according to the below {\em token-transfer rules}:
\begin{definition}[Token-transfer rules]
\label{def04}
\begin{itemize}
\parskip=0pt
\item
	A Type-1 optimal star has $k$ tokens in total and at least one covered satellite.
	We transfer these $k$ tokens to one of the covered satellites.
\item
	A Type-2 optimal star has $x+1$ tokens in total, where $0 \le x \le k-2$ is the number of uncovered satellites, and at least $k-x$ covered satellites.
	We transfer these $x+1$ tokens evenly to $k-x$ of its covered satellites (each receiving $\frac {x+1}{k-x} \le \frac {k-1}2$ tokens).
\item
	A Type-3 optimal star has $y$ tokens in total, where $0 \le y \le k-1$ is the number of uncovered satellites, and at least $k-y$ covered satellites.
	We transfer these $y$ tokens evenly to the center and $k-y$ of its covered satellites (each receiving $\frac y{k-y+1} \le \frac {k-1}2$ tokens).
\end{itemize}
\end{definition}

By Observations~\ref{obs01}(2) and \ref{obs02}, we have the following:
\begin{observation}
\label{obs03}
By the token-transfer rules,
\begin{enumerate}
\parskip=0pt
\item
	the covered vertices receiving tokens in an optimal star and those uncovered vertices induce a connected subgraph inside the optimal star;
	(thus there are always a covered vertex and an uncovered vertex adjacent to each other in the optimal star;)
\item
	a vertex receives $k$ tokens only by the first rule, it is a satellite of an internal $k$-star $S$,
	and no other satellite of $S$ can also receive $k$ tokens; %due to operation {\sc Trade-1} 
\item
	for the satellites of internal $(k + 1)^+$-stars and the centers of internal stars, each receives at most $\frac {k-1}2$ tokens.
\end{enumerate}
\end{observation}
\begin{proof}
The second item can be not so obvious and we provide a proof.
Denote by $v$ this vertex receiving $k$ tokens in the optimal star $S^*$.
Since the center of $S^*$ is uncovered, by Observation~\ref{obs01}(2) $v$ cannot be the center of the internal star $S$.
That is, $v$ is a satellite in $S$.
It then follows from Observation~\ref{obs02} that $S$ is a $k$-star.
Lastly, if another satellite of $S$ receives $k$ tokens, then it must receive these tokens in another optimal star $S^*_1 \ne S^*$.
One then sees that operation {\sc Trade-$1$} would be applicable to trade $S$ with the two $k$-stars centered at the centers of the two optimal stars $S^*$ and $S^*_1$, respectively,
a contradiction to the termination condition of {\sc Approx$[k, \infty)$}.

The third item is a direct consequence of the second item.
\end{proof}

\begin{definition}[Heavy star and its companions]
\label{def05}
An internal $k$-star $S$ is {\em heavy} if a satellite receives $k$ tokens.

Suppose $S$ is heavy.
For another internal $k$-star $S'$, if a satellite $v$ of $S$ and a satellite $v'$ of $S'$ are co-recipients in an optimal star each receiving $\frac {k-1}2$ tokens,
then $S'$ is a {\em companion} for $S$ (see for an illustration in \Cref{fig01}).
\end{definition}

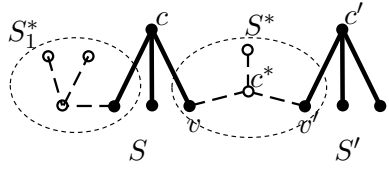
\begin{figure}[ht]
\centering\scalebox{1.0}{
  \setlength{\unitlength}{1bp}%
  \begin{picture}(156.06, 73.52)(0,0)
  \put(0,0){\includegraphics{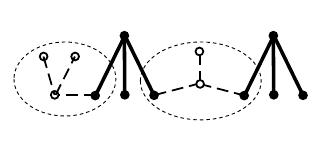}}
  \put(61.20,58.63){\fontsize{11.38}{13.66}\selectfont $c$}
  \put(114.43,18.09){\fontsize{11.38}{13.66}\selectfont $v'$}
  \put(94.49,54.97){\fontsize{11.38}{13.66}\selectfont $S^*$}
  \put(72.77,18.53){\fontsize{11.38}{13.66}\selectfont $v$}
  \put(50.92,8.12){\fontsize{11.38}{13.66}\selectfont $S$}
  \put(128.28,8.12){\fontsize{11.38}{13.66}\selectfont $S'$}
  \put(131.97,58.96){\fontsize{11.38}{13.66}\selectfont $c'$}
  \put(96.66,35.47){\fontsize{11.38}{13.66}\selectfont $c^*$}
  \put(5.67,52.80){\fontsize{11.38}{13.66}\selectfont $S^*_1$}
  \end{picture}%
}
\captionsetup{width=.95\textwidth}
\caption{An illustration of a companion star $S'$ for a heavy star $S$, of which the edges are thick solid:
	One satellite of $S$ receives $k$ tokens, here $k = 3$;
	another satellite $v$ of $S$ and a satellite $v'$ of $S'$ each receives $\frac {k-1}2$ tokens in the optimal star $S^*$ (dotted circled).\label{fig01}}
\end{figure}

\begin{lemma}
\label{lemma01}
\begin{enumerate}
\parskip=0pt
\item
	The center of a heavy star receives at most $\frac {k-2}3$ tokens;
\item
	the companion star of a heavy star is not heavy;
\item
	the center of a companion star receives at most $\frac {k-2}3$ tokens;
\item
	no two heavy stars share a common companion.	
\end{enumerate}
\end{lemma}
\begin{proof}
Let $S$ be a heavy star, i.e., one of its satellites receives $k$ tokens in an optimal star $S^*_1$.
By Observation~\ref{obs03}(2) no other satellite can receive $k$ tokens, but at most $\frac {k-1}2$ tokens.
If the center $c$ of $S$ receives $\frac {k-1}2$ tokens, then by Observation~\ref{obs03}(1) the tokens are transferred to $c$ in an optimal star $S^*_2 \ne S^*_1$ by the third rule,
with its co-recipient $u$ adjacent to the $k-1$ uncovered vertices in $S^*_2$.
It follows that $u$ is a satellite in an internal star $S'$, but then operation {\sc Trade-$2$} would be applicable to trade $S$ and $S'$
with a $k$-star centered at the center of $S^*_1$ and a $(k+1)$-star centered at $u$ if $S'$ is a $k$-star (including the degenerate case where $S = S'$),
or otherwise a $k$-star centered at $u$ and a revised star $S'$ by removing $u$.
This proves item 1.

Let $S'$ be a companion star for $S$, and assume a satellite $v$ of $S$ and a satellite $v'$ of $S'$ are co-recipients in an optimal star $S^*$ each receiving $\frac {k-1}2$ tokens
(see for an illustration in \Cref{fig01}).
If $S'$ is heavy and assume one of its satellites receives $k$ tokens in an optimal star $S^*_2 \ne S^*_1$,
then operation {\sc Trade-2} would be applicable to trade $S$ and $S'$
with three $k$-stars centered at the centers of $S^*_1$, $S^*$ and $S^*_2$, respectively.

Next, if the center $c'$ of $S'$ receives $\frac {k-1}2$ tokens, then the same the tokens are transferred to $c'$ in an optimal star $S^*_2 \ne S^*$ by the third rule,
with its co-recipient $u'$ adjacent to the $k-1$ uncovered vertices in $S^*_2$.
It follows that $u'$ is a satellite in an internal star $S''$, but then operation {\sc Trade-$3$} would be applicable to trade $S$, $S'$ and $S''$
with two $k$-stars centered at the centers of $S^*_1$ and $S^*$, respectively,
and a $(k+1)$-star centered at $u'$ if $S''$ is a $k$-star (including the degenerate case where $S''$ collides into $S$ or $S'$),
or otherwise a $k$-star centered at $u'$ and a revised star $S''$ by removing $u'$.

Lastly, if $S'$ is also a companion for another heavy star $S''$, and assume one satellite of $S''$ receives $k$ tokens in an optimal star $S^*_2 \ne S^*_1$,
a satellite of $S'$ and a satellite of $S''$ are co-recipients each receiving $\frac {k-1}2$ tokens in an optimal star $S^*_3 \ne S^*$,
then operation {\sc Trade-3} would be applicable to trade $S$, $S'$ and $S''$
with four $k^+$-stars centered at the centers of $S^*_1$, $S^*_2$, $S^*$ and $S^*_3$, respectively.
This proves the lemma.
\end{proof}

In the proof of the following main theorem in this section, we examine the total tokens received by (all the vertices of) each internal star.

\begin{theorem}
\label{thm01}
The algorithm {\sc Approx$[k, \infty)$} is an $O(|V|^{9.5})$-time $\frac {k+1}2$-approximation algorithm for the {\sc $[k, \infty)$-Star Packing} problem, for any $k \ge 3$.
\end{theorem}
\begin{proof}
We have mentioned earlier that the most expensive operation is {\sc Trade-3}, which takes $O(|V|^{8.5})$ time.
Since every operation increases the number of covered vertices by at least one, the total number of operations is $O(|V|)$, leading to an overall running time of $O(|V|^{9.5})$.

For each internal star $S \in \mcP$, the total tokens received by all its vertices through token-transfer process, denoted as $\tau(S)$,
is regarded as the number of uncovered vertices in $V(\mcQ)$ that are distributed to the star $S$.
As every token is transferred to some covered vertices, we have $|V(\mcQ)| \le |V(\mcP)| + \sum_{S \in \mcP} \tau(S)$.

For an internal $(k + 1)^+$-star $S$, by Observation~\ref{obs03}(3) every vertex receives no more than $\frac {k-1}2$ tokens, and thus $\tau(S) \le |V(S)| \cdot \frac {k-1}2$.

Consider a heavy (internal $k$-)star $S$, i.e., one of its satellites receives $k$ tokens in an optimal star $S^*_1$.
By \Cref{lemma01}, its center $c$ receives no more than $\frac {k-2}3$ tokens.
If a satellite $v$ of $S$ receives $\frac {k-1}2$ tokens in an optimal star $S^*$,
then $v$ cannot be the center of $S^*$ since otherwise operation {\sc Trade-$1$} would be applicable to trade $S$
with two $k^+$-stars centered at the center of $S^*_1$ and $v$, respectively.
It follows that the co-recipient $u$ in $S^*$ is adjacent to an uncovered vertex in $S^*$ and thus $u$ is a satellite in an internal star $S'$.
Furthermore, $S' \ne S$ since otherwise operation {\sc Trade-1} would be applicable to trade $S$
with two $k^+$-stars centered at the centers of $S^*_1$ and $S^*$, respectively.
On the other hand, $S'$ is not a $(k+1)^+$-star or otherwise operation {\sc Trade-2} would be applicable to trade $S$ and $S'$
with two $k$-stars centered at the centers of $S^*_1$ and $S^*$, respectively, and a revised star $S'$ by removing $u$.
We conclude that $S'$ is a $k$-star, and thus it is a companion for $S$ (and $u$ cannot be the center of $S^*$ either, see for an illustration in \cref{fig01}).
This also proves that when $j$ satellites of $S$ each receives $\frac {k-1}2$ tokens, $j$ distinct companion stars for $S$ are located, denoted as $S_1, S_2, \ldots, S_j$.
As a collection of internal $k$-stars, by \Cref{lemma01} we have
\begin{eqnarray}
\tau(\{S, S_1, S_2, \ldots, S_j\}) &\le &\left(k + j \cdot \frac {k-1}2 + (k-j) \cdot \frac {k-2}3 \right) 	+ j \cdot \left(k \cdot \frac {k-1}2 + \frac {k-2}3 \right)\nonumber\\
	&= 	&(k + 1) \cdot \frac k3 + j (k + 1) \cdot \frac {k-1}2\nonumber\\
 	&\le &(j + 1) (k + 1) \cdot \frac {k-1}2, \mbox{ for any } k \ge 3.\label{eq01}
\end{eqnarray}
That is, on average each star in this collection receives no more than $(k + 1) \cdot \frac {k-1}2$ tokens.

For every other internal $k$-star $S$ which is neither heavy nor a companion,
the total token received by $S$ is also no more than $(k + 1) \cdot \frac {k-1}2$.

In summary,
\[
\frac {|V(\mcQ)|}{|V(\mcP)|} \le 1 + \frac {\sum_{S \in \mcP} \tau(S)}{|V(\mcP)|} \le 1 + \max_{S \in \mcP} \frac {\tau(S)}{|V(S)|} = 1 + \frac {k-1}2 = \frac {k+1}2.
\]
This proves the approximation ratio.
\end{proof}

%\newpage
\subsection{A revised algorithm for $k = 2$ and the performance analysis}
%--------------------------------------------------------------------------------------------------
Note that to prove the ratio of $\frac {k+1}2$ for the algorithm {\sc Approx$[k, \infty)$} in the above \Cref{thm01},
we require $\frac k3 \le \frac {k-1}2$ in Eq.~(\ref{eq01}) which is equivalent to $k \ge 3$.
Nevertheless, the algorithm {\sc Approx$[k, \infty)$} works for $k = 2$,
just that the performance analysis shows a guarantee of only $1 + \frac k3 = \frac 53$ but not $1 + \frac {k-1}2 = \frac 32$.
One could proceed with a refined analysis to show that {\sc Approx$[k, \infty)$} is a $\frac{19}{12}$-approximation algorithm for {\sc $[2, \infty)$-Star Packing},
which we choose to skip but add one novel operation to become {\sc Approx$[2, \infty)$} and prove it to have a much better approximation ratio of $\frac 43$.

We first present an augmenting configuration, which is explored in the local improvement operation.
\begin{definition}[Triplet]
\label{def06}
A {\em triplet} $(v_0, S, v)$ consists of a satellite $v_0$ in an internal $3^+$-star, an internal $2$-star $S$ centered at $c$,
and a satellite $v \notin V(S)$ adjacent to $c$ in the input graph.
\end{definition}

\begin{definition}[Augmenting configuration, abbreviated as AC]
\label{def07}
Given a triplet $(v_0, S, v)$ where $v_0$ is a satellite of $S_0$,
if there is a sequence of $j \ge 0$ internal $2$-stars $S_1$, $S_2$, $\ldots$, $S_j$ such that
$v_0$ is adjacent to the center of $S_1$, a satellite of $S_i$ is adjacent to the center of $S_{i+1}$ for every $i = 1, 2, \ldots, j-1$, and $v$ is a satellite of $S_j$,
then we say this sequence of stars $\langle S_0, S_1, \ldots, S_j, S\rangle$ form into an {\em augmenting configuration} abbreviated as {\em AC} for the triplet $(v_0, S, v)$,
see for an illustration in \Cref{fig02}.
({\em Comment:} In the extreme case where $j = 0$, $v$ collides into $v_0$.)
\end{definition}

\begin{figure}[ht]
\centering\scalebox{1.0}{
  \setlength{\unitlength}{1bp}%
  \begin{picture}(228.74, 64.93)(0,0)
  \put(0,0){\includegraphics{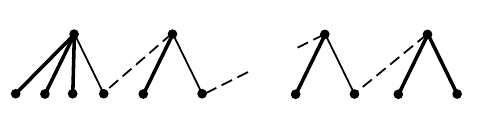}}
  \put(84.28,50.38){\fontsize{11.38}{13.66}\selectfont $c_1$}
  \put(121.88,34.46){\fontsize{11.38}{13.66}\selectfont $\ldots$}
  \put(206.69,50.25){\fontsize{11.38}{13.66}\selectfont $c$}
  \put(157.40,50.25){\fontsize{11.38}{13.66}\selectfont $c_j$}
  \put(74.33,8.12){\fontsize{11.38}{13.66}\selectfont $S_1$}
  \put(142.69,8.12){\fontsize{11.38}{13.66}\selectfont $S_j$}
  \put(198.62,8.12){\fontsize{11.38}{13.66}\selectfont $S$}
  \put(18.40,8.12){\fontsize{11.38}{13.66}\selectfont $S_0$}
  \put(44.53,8.12){\fontsize{11.38}{13.66}\selectfont $v_0$}
  \put(167.69,8.12){\fontsize{11.38}{13.66}\selectfont $v$}
  \end{picture}%
}
\captionsetup{width=.95\textwidth}
\caption{An augmenting configuration for the triplet $(v_0, S, v)$,
	from the internal $3^+$-star $S_0$ to the internal $2$-star $S$ via a sequence of internal $2$-stars $S_1, S_2, \ldots, S_j$,
	where the solid edges are in the internal stars while the dashed edges are in the input graph but not in these stars.
	Swapping the dashed edges into the stars with the thin solid edges expands $S$ into a $3$-star,
	with the effect of decreasing the number of satellites of $S_0$ by one and increasing the number of satellites of $S$ by one,
	or we can imagine it as ``moving'' a satellite of $S_0$ to $S$.\label{fig02}}
\end{figure}
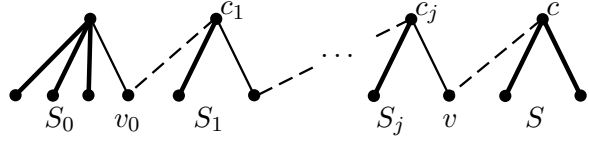

For the current solution $\mcP$, we can implicitly construct a digraph $\mcH$ on the internal stars,
such that if a star $S_1$ has a satellite adjacent to the center of a $2$-star $S_2$ then there is an arc from $S_1$ to $S_2$.
Existence of an AC for the triplet $(v_0, S, v)$ in \Cref{def07} is equivalent to the existence of a directed path from $S_0$ to $S$ in $\mcH$,
and thus can be determined in $O(|V|^2)$ time via a breadth-first-search traversal in $\mcH$.
Furthermore, in the sequel, by an AC for the triplet $(v_0, S, v)$ we always mean a {\em shortest} one, i.e., the sequence consists of the smallest number of stars.
With the AC in \Cref{def07} (see for an illustration in \Cref{fig02}), one can swap the satellite $v_0$ of $S_0$ into $S_1$ with the satellite of $S_1$ adjacent to the center of $S_2$,
which in turn is swapped into $S_2$ with the satellite of $S_2$ adjacent to the center of $S_3$, and so on,
and lastly adds the satellite $v$ of $S_j$ into $S$ to {\em expand} $S$ into a $3$-star.
Such a process along the AC does not change the covered vertices,
but decrease the number of satellites of $S_0$ by one and increase the number of satellites of $S$ by one,
or we can imagine it as ``{\em moving}'' a satellite of $S_0$ to $S$.
Note that there are $O(|V|^3)$ triplets.
We define the {\sc Trade-along-AC(s)} operation below.

\begin{definition}[Operation {\sc Trade-along-AC(s)}] 
\label{def08}
\begin{itemize}
\parskip=0pt
\item
Given a triplet $(v_0, S, v)$ where $v_0$ is a satellite of $S_0$,
let $S'_0$ denote the star obtained from $S_0$ by removing $v_0$, $S'$ denote the $3$-star obtained from $S$ by adding $v$, and $\mcP' = \mcP \setminus \{S_0, S\} \cup \{S'_0, S'\}$.
Let $\mcS$ denote a collection consisting of $S'$ and $i-1$ other stars of $\mcP'$, for some $1 \le i \le 3$.
If a {\sc Trade-$i$} operation is applicable on $\mcS$, and there exists an AC for the triplet $(v_0, S, v)$ with all the intermediate $2$-stars in $\mcP' \setminus \mcS$,
then the operation first modifies $\mcP$ into $\mcP'$ using the AC and then executes the {\sc Trade-$i$} operation on $\mcS$.

({\em Comment}: The order of locating $\mcS$ first and then the existence of AC is important, because there can be exponentially many ACs.
Since $S'$ is fixed for the triplet $(v_0, S, v)$, finding all applicable {\sc Trade-$i$} operations can be done in $O(|V|^{(i-1) + (i+1) + 1.5}) = O(|V|^{2i+1.5})$ time.
For an operation {\sc Trade-$i$} applicable on $\mcS$,
finding an AC with all the intermediate $2$-stars in $\mcP' \setminus \mcS$, then expanding $S$ to the $3$-star $S'$ along the AC, can be done in $O(|V|^2)$ time.
If no {\sc Trade-$i$} operation is applicable or no such AC exists, then the operation {\sc Trade-along-AC} is not applicable on the triplet $(v_0, S, v)$.
Therefore, finding and executing an operation {\sc Trade-along-AC} takes $O(|V|^{3 + (2i+1.5) + 2}) = O(|V|^{2i+6.5})$ time.)

\item
Given two triplets $(v_0, S_1, v)$ and $(u_0, S_2, u)$ where $u_0 \ne v_0$ are satellites of two $3^+$-stars $S_0$ and $T_0$ or are two satellites of a $4^+$-star $S_0 = T_0$,
$S_1 \ne S_2$, and $v \ne u$,
let $S'_0$ ($T'_0$, resp.) denote the star obtained from $S_0$ ($T_0$, resp.) by removing $v_0$ ($u_0$, resp.),
$S'_1$ ($S'_2$, resp.) denote the $3$-star obtained from $S_1$ ($S_2$, resp.) by adding $v$ ($u$, resp.), and
$\mcP' = \mcP \setminus \{S_0, T_0, S_1, S_2\} \cup \{S'_0, T'_0, S'_1, S'_2\}$.
Let $\mcS$ denote a collection consisting of $S'_1$, $S'_2$ and $i-2$ other internal stars of $\mcP'$, for some $2 \le i \le 3$.
If a {\sc Trade-$i$} operation is applicable on $\mcS$, and there exists edge-disjoint ACs for the triplets $(v_0, S_1, v)$ and $(u_0, S_2, u)$
with all the intermediate $2$-stars in $\mcP' \setminus \mcS$,
then the operation first modifies $\mcP$ into $\mcP'$ using the ACs and then executes the {\sc Trade-$i$} operation on $\mcS$.

({\em Comment}: When $S_0 = T_0$, $S'_0 = T'_0$ is the star obtained by removing $v_0$ and $u_0$ from $S_0$.
Finding all applicable {\sc Trade-$i$} operations can be done in $O(|V|^{(i-2) + (i+1) + 1.5}) = O(|V|^{2i + 0.5})$ time.
For an operation {\sc Trade-$i$} applicable on $\mcS$,
finding two edge-disjoint ACs with all the intermediate $2$-stars in $\mcP' \setminus \mcS$, then expanding $S_1$ and $S_2$ to the $3$-stars along the ACs, can be done in $O(|V|^2)$ time.
If no {\sc Trade-$i$} operation is applicable or no such ACs exist, then the operation is not applicable on the pair of triplets $(v_0, S_1, v)$ and $(u_0, S_2, u)$.
Therefore, finding and executing an operation {\sc Trade-along-ACs} takes $O(|V|^{6 + (2i+0.5) + 2}) = O(|V|^{2i+8.5})$ time.)
\end{itemize}
\end{definition}

One sees that finding and executing an operation {\sc Trade-along-ACs} has a high time complexity of $O(|V|^{2i+8.5}) = O(|V|^{14.5})$, when $i = 3$.
Note that some special cases of {\sc Trade-along-AC(s)} have been covered by {\sc Trade-$i$} operations, such as the case where $v_0$ collides into $v$.
Therefore, {\sc Trade-along-AC(s)} is regarded as an extended {\sc Trade-$i$} operation, and is also referred to as {\sc Trade}.
Using the {\sc Collect} and {\sc Trade} operations, our algorithm {\sc Approx$[2, \infty)$} repeatedly applies them whenever possible to update the solution.
When none of them is applicable, the algorithm returns the achieved solution $\mcP$ as the final solution to the {\sc $[2, \infty)$-Star Packing} problem.

We remark that the previous best $\frac 32$-approximation algorithm~\cite{HZC26} employs
a special case of operation {\sc Trade-$3$}~\cite[Definition 6 of {\sc Pull-by-$(2, 2, 2)$}]{HZC26}
where all three to-be-traded internal stars are all $2$-stars and there are at least two $2$-stars among the newly extracted ones,
besides those operations for the general $k \ge 3$.
Augmenting structures have been designed and used in the polynomial time exact algorithms by Hell and Kirkpatrick~\cite{HK86} for {\sc $[1, \ell]$-Star Packing}.
Here we explore the novel augmenting configurations in the local improvement operation {\sc Trade-along-AC(s)}.

%--------------------------------------------------------------------------------------------------
To analyze the performance of {\sc Approx$[2, \infty)$}, we similarly fix an optimal $2^+$-star packing denoted as $\mcQ$ for discussion,
and further choose such a $\mcQ$ that contains the maximum number of covered vertices among all optimal $2^+$-star packings
(that is, $\mcQ$ has the {\em maximum overlap} with the computed solution $\mcP$).
We categorize the optimal stars into three types exactly the same as \Cref{def03}, but this time using $k = 2$;
we then associate a token with each uncovered vertex of $V(\mcQ)$,
and {\em transfer} the tokens in each optimal star, this time in wholes,
to some of the covered vertices in the same optimal star according to the following {\em token-transfer rules} adjusted from \Cref{def04}:
\begin{definition}[Token-transfer rules]
\label{def09}
\begin{itemize}
\parskip=0pt
\item
	A Type-1 optimal star has two tokens and its all but one satellite are covered.
	We transfer the two tokens to a covered satellite.
\item
	A Type-2 optimal star has one token and its satellites are all covered.
	We transfer the token to a satellite.
\item
	A Type-3 optimal star has at most one token, and if exists, then all but one satellite are covered and we transfer the token to the center.
\end{itemize}
\end{definition}

By Observations~\ref{obs01} and \ref{obs02}, we have the following analogous to Observation~\ref{obs03}.
\begin{observation}
\label{obs04}
By the token-transfer rules,
\begin{enumerate}
\parskip=0pt
\item
	each vertex receiving tokens is adjacent to an uncovered vertex, and thus the center of an internal star receives no token;
\item
	only a satellite of an internal $2$-star can receive two tokens, and in this case the other satellite of the $2$-star receives no token;
\item
	except stated in the above second item, a satellite of any other internal star receives at most one token.
\end{enumerate}
\end{observation}
\begin{proof}
The second item can be not too obvious and we provide a proof.
Denote by $v_1$ the satellite of a $2$-star $S$ which receives the two tokens in an optimal star $S^*_1$, and denote the other satellite of $S$ by $v_2$.
If $v_2$ receives the two tokens in another optimal star $S^*_2$, then operation {\sc Trade-$1$} would be applicable to trade $S$
with the two $2^+$-stars centered at the centers of $S^*_1$ and $S^*_2$, respectively.
If $v_2$ receives the one token in another optimal star $S^*_2$, then operation {\sc Trade-$1$} would be applicable to trade $S$
with the two $2^+$-stars centered at the center of $S^*_1$ and $v_2$, respectively.
Either way it is a contradiction to the termination condition of {\sc Approx$[2, \infty)$}.
\end{proof}

In the proof of the following main theorem in this section, we examine the average number of tokens received by all the internal stars.
For each internal star $S$, the total tokens received by all its vertices is denoted as $\tau(S)$.
We extend this notation to a collection $\mcS$ of internal stars that $\tau(\mcS)$ denotes the total tokens they receive.

\begin{definition}[Heavy $2$-star]
\label{def10}
An internal $2$-star $S$ is {\em heavy} if $\tau(S) = 2$.
\end{definition}

\begin{theorem}
\label{thm02}
The algorithm {\sc Approx$[2, \infty)$} is an $O(|V|^{15.5})$-time $\frac 43$-approximation algorithm for the {\sc $[2, \infty)$-Star Packing} problem.
\end{theorem}
\begin{proof}
We distinguish $2$-stars and $3^+$-stars.

{\bf Case 1: $2$-stars}.
Consider an internal $2$-star $S$.
By Observation~\ref{obs04}, the two satellites of $S$ receive at most two tokens in total and there are two cases where $\tau(S) = 2$.

In the first case, a satellite $v_1$ of $S$ receives two tokens in an optimal star $S^*_1$,
while the other satellite $v_2$ receives no token (see for an illustration in \Cref{fig03}, the left half).
By the assumption of $\mcQ$, at least one of $v_2$ and $c$ is in some optimal star
--- because otherwise $S^*_1$ is a $2$-star and then $S^*_1$ can be replaced by $S$ to strictly increase the number of covered vertices in $\mcQ$.
Assume $v_2$ is in an optimal star $S^*_2$ (the case where $c$ is in $S^*_2$ is discussed similarly).
We claim that $S^*_1 \ne S^*_2$, since otherwise operation {\sc Trade-$1$} would be applicable to trade $S$ with $S^*_1$.
It follows that $v_2$ is adjacent to a covered vertex $u_1 \notin V(S)$ in $S^*_2$%
\footnote{Proof of the existence of $u_1$: If $v_2$ is a satellite of $S^*_2$,
	then $u_1$ is the center of $S^*_2$ and it is covered as otherwise operation {\sc Trade-$1$} would be applicable to trade $S$
	with $S^*_1$ and a $2$-star centered at $v_2$;
	if $u_1$ happens to be $c$, then we switch the argument on $c$ and
	by the fact that $c$ is not adjacent to any uncovered vertex we can set $u_1$ to be a covered satellite in $S^*_2$ other than $v_2$.
	If $v_2$ is the center of $S^*_2$, then $S^*_2$ is Type-3 and has no uncovered satellite, and thus we can set $u_1$ to be a covered satellite of $S^*_2$ other than $c$.},
i.e., $u_1$ is in an internal star $S' \ne S$.
We claim that $u_1$ is a satellite in $S'$, $S'$ is a $2$-star, and $\tau(S') = 0$, since otherwise operation {\sc Trade-$2$} would be applicable to trade $S$ and $S'$
with $S^*_1$ and one or two stars centered at $v_2, u_1$ and/or the center or the other satellite of $S'$.
This star $S'$ is called a {\em companion} of $S$, identified through the adjacency $v_2$-$u_1$ in $S^*_2$.
(One sees that there can be multiple companions for $S$, for example, when $S^*_2$ is a $3^+$-star centered at $v_2$ or $c$ is also in some optimal star $S^*_3$.
Nevertheless, one such companion is sufficient for $S$.)

\begin{figure}[ht]
\centering\scalebox{1.0}{
  \setlength{\unitlength}{1bp}%
  \begin{picture}(270.92, 88.20)(0,0)
  \put(0,0){\includegraphics{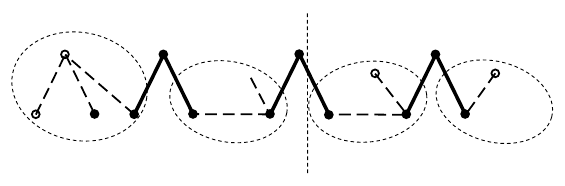}}
  \put(79.89,64.65){\fontsize{11.38}{13.66}\selectfont $c$}
  \put(135.06,34.44){\fontsize{11.38}{13.66}\selectfont $u_1$}
  \put(72.58,22.40){\fontsize{11.38}{13.66}\selectfont $S$}
  \put(136.66,21.03){\fontsize{11.38}{13.66}\selectfont $S'$}
  \put(19.46,12.90){\fontsize{11.38}{13.66}\selectfont $S^*_1$}
  \put(58.77,40.90){\fontsize{11.38}{13.66}\selectfont $v_1$}
  \put(93.08,40.90){\fontsize{11.38}{13.66}\selectfont $v_2$}
  \put(93.36,12.90){\fontsize{11.38}{13.66}\selectfont $S^*_2$}
  \put(155.38,41.89){\fontsize{11.38}{13.66}\selectfont $u_2$}
  \put(139.85,64.65){\fontsize{11.38}{13.66}\selectfont $c'$}
  \end{picture}%
}
\captionsetup{width=.95\textwidth}
\caption{Finding a companion $2$-star $S'$ for a heavy $2$-star $S$.
	The left half illustrates the first case where one satellite receives two tokens, and the right illustrates the second case where each of the two satellites receives one token.
	Solid edges are in the internal stars and dashed edges are in the optimal stars.\label{fig03}}
\end{figure}
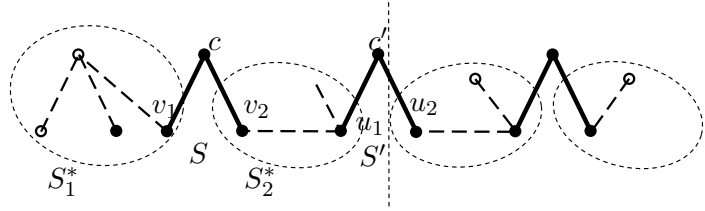

In the second case, each of the two satellites $v_1$ and $v_2$ of $S$ receives one token, in two distinct optimal stars $S^*_1$ and $S^*_2$, respectively.
Since each of $S^*_1$ and $S^*_2$ contains at least two covered vertices,
at least one of them, say $S^*_2$, contains a covered vertex denoted $u_2 \notin V(S)$ (see for an illustration in \Cref{fig03}, the right half),
i.e., $u_2$ is in another internal star $S' \ne S$.
We claim that $u_2$ is a satellite in $S'$, $S'$ is a $2$-star, and $\tau(S') = 0$, since otherwise operation {\sc Trade-$2$} would be applicable to trade $S$ and $S'$
with a $2^+$-star centered at $v_1$ and one or two stars centered at $v_2, u_2$ and/or the center or the other satellite of $S'$.
This star $S'$ is a companion of $S$, identified through the adjacency $v_2$-$u_2$ in $S^*_2$.
Again, one sees that there can be more companions for $S$, for example, when $S^*_2$ is a $3^+$-star.
Specifically, when $S^*_1$ also contains a covered vertex denoted $u_1 \notin V(S)$ (this happens if the center $c$ of $S$ is not in $S^*_1$),
the same argument as for $S^*_2$ says that a distinct companion $S''$ of $S$ is identified through the adjacency $v_1$-$u_1$ in $S^*_1$.
We identify one such companion, which is sufficient, through each of $S^*_1$ and $S^*_2$ for $S$, when the optimal star contains a vertex not covered by $S$.

Lastly, if an internal $2$-star $S'$ is identified as a companion for a heavy star $S$ through an adjacency $v_2$-$u_1$ in the first case above,
then $S$ and $S'$ can be replaced by two stars centered at the center of $S^*_1$ and $v_2$, respectively (see for an illustration in \Cref{fig03}, the left half);
this way the two uncovered vertices with their token transferred to the satellite $v_1$ of $S$ can be swapped to become covered
while the center $c'$ and the other satellite $u_2$ of $S'$ become uncovered.
If an internal $2$-star $S'$ is identified as a companion for a heavy star $S$ through an adjacency $v_1$-$u_1$ in the second case above,
then $S$ and $S'$ can be replaced by two stars centered at $v_1$ and $v_2$, respectively (see for an illustration in \Cref{fig03}, the right half);
this way the two uncovered vertices with their token transferred to the satellites $v_1$ and $v_2$ of $S$ can be swapped to become covered
while the center $c'$ and the other satellite $u_1$ of $S'$ become uncovered.
This implies that $S'$ is not a companion for another heavy star $S''$,
since otherwise operation {\sc Trade-$3$} would be applicable to trade $S, S'$ and $S''$ to cover at least three more vertices
(see for an illustration in \Cref{fig03}, the two halves cannot happen simultaneously).
That is, no heavy $2$-stars share any companion.

{\bf Case 2: $3^+$-stars}.
Consider an internal $3^+$-star $S$.
By Observation~\ref{obs04}, every satellite of $S$ receives at most one token.
Let $v$ denote a satellite receiving one token, $x$ denote the adjacent uncovered vertex that transfers its token to $v$ in an optimal star $S^*$
(see for illustrations in \Cref{fig04}).

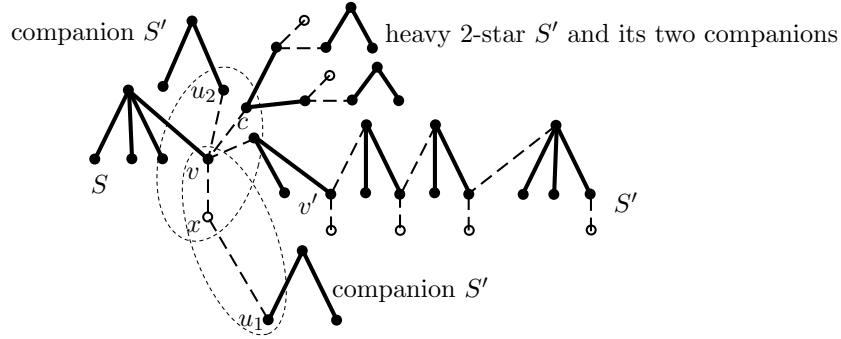
\begin{figure}[ht]
\centering\scalebox{0.9}{
  \setlength{\unitlength}{1bp}%
  \begin{picture}(314.65, 149.66)(0,0)
  \put(0,0){\includegraphics{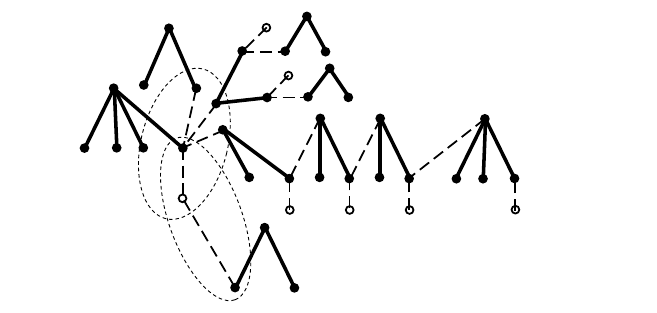}}
  \put(161.59,127.84){\fontsize{11.38}{13.66}\selectfont heavy $2$-star $S'$ and its two companions}
  \put(38.86,64.69){\fontsize{11.38}{13.66}\selectfont $S$}
  \put(78.63,70.65){\fontsize{11.38}{13.66}\selectfont $v$}
  \put(139.45,21.42){\fontsize{11.38}{13.66}\selectfont companion $S'$}
  \put(5.67,128.58){\fontsize{11.38}{13.66}\selectfont companion $S'$}
  \put(256.89,56.70){\fontsize{11.38}{13.66}\selectfont $S'$}
  \put(78.81,48.61){\fontsize{11.38}{13.66}\selectfont $x$}
  \put(80.45,104.48){\fontsize{11.38}{13.66}\selectfont $u_2$}
  \put(99.72,91.46){\fontsize{11.38}{13.66}\selectfont $c$}
  \put(100.47,9.39){\fontsize{11.38}{13.66}\selectfont $u_1$}
  \put(125.35,56.12){\fontsize{11.38}{13.66}\selectfont $v'$}
  \end{picture}%
}
\captionsetup{width=.95\textwidth}
\caption{All possibilities of exploring from a satellite $v$ of an internal $3^+$-star $S$,
	where $v$ is adjacent to an uncovered vertex $x$ and receives its token in an optimal star $S^*$.
	The lower dotted circle shows the case $S^*$ centering at $x$, leading to a companion $2$-star $S'$ for $S$.
	The upper dotted circle shows the case $S^*$ centering at $v$:
	A covered satellite of $S^*$ could be the center of an internal $2$-star $S'$ receiving no token and thus becomes a companion for $S$;
	or one could be the center of a heavy internal $2$-star $S'$ which has at least two companions;
	or one could be the center of an internal $3^+$-star or an internal $2$-star receiving no token;
	or lastly one could be the center of an internal $2$-star receiving one token and then a satellite $v'$ receiving the token takes up the role of $v$ to repeat the exploration.
	Solid edges are in the internal stars and dashed edges are in the optimal stars.\label{fig04}}
\end{figure}

Firstly, if $x$ is the center of $S^*$ (lower part of \Cref{fig04}), i.e., $S^*$ is Type-2,
then every satellite $u_1$ of $S^*$ is a satellite in an internal $2$-star $S'$ as otherwise operation {\sc Trade-$2$} would be applicable to trade $S$ and $S'$
with a star centered at the center of $S$ and one or two stars centered at $x$ and/or the center of $S'$;
the other satellite of $S'$ receives no token either, i.e., $\tau(S') = 0$, since otherwise operation {\sc Trade-$2$} would be applicable to trade $S$ and $S'$
with a star centered at the center of $S$ and two stars centered at $x$ and this satellite of $S'$.
We call $S'$ a companion for $S$ identified through the path $v$-$x$-$u_1$ in $S^*$.
Note that being a companion of $S$, $S$ and $S'$ can be replaced by two stars centered at the center of $S$ and $x$, respectively
(see for an illustration in \Cref{fig04}, the lower part);
this way the uncovered vertex $x$ is swapped to be covered together with $u_1$ while the center and the other satellite of $S'$ become uncovered.
Consequently, $S'$ cannot be a companion for any other internal ($3^+$- or $2$-) star $S''$,
since otherwise operation {\sc Trade-$3$} would be applicable to trade $S, S'$ and $S''$ to cover at least one more vertex.
(We note again that there can be multiple companions for $S$, for example, when $S^*$ is a $3^+$-star.
Nevertheless, identifying one such companion is sufficient.)

Consider the next case where $v$ is the center of $S^*$, i.e., $S^*$ is Type-3, and let $u_2$ be a covered satellite of $S^*$.
If $u_2$ is a satellite in an internal star $S'$ (upper left part of \Cref{fig04}),
then $S'$ is a $2$-star and $\tau(S') = 0$ since otherwise operation {\sc Trade-$2$} would be applicable to trade $S$ and $S'$
with two stars centered at the centers of $S$ and $v$, respectively, and a star centered at the center or the other satellite of $S'$.
We call $S'$ a companion for $S$ identified through the adjacency $v$-$u_2$ in $S^*$.
We remark that such a companion has the same property as proven in the last paragraph, that it cannot be a companion for another internal star.
(We note again that there can be multiple companions for $S$, for example, when $S^*$ has more such covered satellites as $u_2$.
Nevertheless, identifying one such companion is sufficient.)

If the covered satellite of $S^*$ is the center $c$ of an internal $3^+$-star $S'$,
then the adjacency $v$-$c$ in $S^*$ is called the {\em critical edge} from (the center of) $S^*$ into (the center of) $S'$.
One sees that, for every internal $3^+$-star such as $S'$, there is at most one critical edge into it.

If the covered satellite of $S^*$ is the center $c$ of an internal $2$-star $S'$ (upper/middle right part of \Cref{fig04}),
then there are three possibilities based on the tokens $S'$ receives:
\begin{itemize}
\parskip=0pt
\item[1)] $\tau(S') = 2$.
	In this scenario, the two satellites of $S'$ cannot be in the same optimal star $S^*_1$, or one satellite receives two tokens in $S^*_1$,
	since otherwise operation {\sc Trade-$2$} would be applicable to trade $S$ and $S'$ with three stars centered at the centers of $S$, $S^*$ and $S^*_1$, respectively,
	to cover at least two more vertices.
	The edge $v$-$c$ is called the {\em critical edge} from (the center of) $S^*$ into (the center of) $S'$.
	Since each of the two satellites of $S'$ receives a token, $S'$ has at least two companions by the second case of a heavy star ({\bf Case 1} above).
\item[2)] $\tau(S') = 0$.
	In this scenario the edge $v$-$c$ is also called the {\em critical edge} from (the center of) $S^*$ into (the center of) $S'$.
\item[3)] $\tau(S') = 1$ and assume the satellite $v'$ of $S'$ receives a token in an optimal star $S^*_1$.
	The other satellite of $S'$ cannot be in the same optimal star $S^*_1$,
	since otherwise operation {\sc Trade-$2$} would be applicable to trade $S$ and $S'$ with three stars centered at the centers of $S$, $S^*$ and $S^*_1$, respectively,
	to cover at least two more vertices.

	One sees that in this scenario, the satellite $v$ of the $3^+$-star $S$ can be {\em cut off} $S$ and then {\em add to} the star $S'$ to {\em expand} $S'$ into a $3$-star.
	This is exactly an augmenting configuration (AC) found for the triplet $(v, S', v)$, which can readily be used to {\em move} a satellite of $S$ to $S'$.
	Therefore, we can repeat all the above arguments on $v$ to $v'$.
	This way, due to {\sc Trade-along-AC} being not applicable,
	either a companion is identified for $S$, or a critical edge into another internal star is identified,
	or lastly another internal $2$-star $S''$ with $\tau(S'') = 1$ is found to extend the AC and further repeat the arguments.%
\footnote{For the completeness of the same argument applicable repeatedly, assume there are multiple internal $2$-stars on the AC from $S$ to $S'$.
	Let $x'$ denote the uncovered vertex of which the token is transferred to $v'$ in an optimal star $S^*_1$.
	If $x'$ is the center of $S^*_1$, i.e., $S^*_1$ is Type-2,
	then every satellite $u'_1$ of $S^*_1$ is a satellite in an internal $2$-star $S'_1$
	as otherwise operation {\sc Trade-along-AC} would be applicable in which the operation {\sc Trade-$2$} trades $S'$ and $S'_1$
	with the stars centered at the center of $S'$, $x'$ and $u'_1$, respectively, to cover at least one more vertex.
	If the other satellite of $S'_1$ receives one token and $S'_1$ is on the AC, then at the time the AC extended to $S'_1$,
	{\sc Trade-along-AC} was applicable in which the operation {\sc Trade-$2$} traded $S'_1$ and $S'$
	with the stars centered at the center of $S'$, $x'$ and this satellite of $S'_1$, respectively, to cover at least two more vertices.
	This proves that if $S'_1$ in not on the AC, i.e., newly discovered,
	then $\tau(S'_1) = 0$ and $S'$ is called a companion for $S$ identified through the path $v$-$c$-$\ldots$-$v'$-$x'$-$u'_1$.
	The other case where $v'$ is the center of the optimal star $S^*_1$ is discussed similarly, in which the contradiction made by {\sc Trade-$i$} is replaced by {\sc Trade-along-AC}.}
	Also due to {\sc Trade-along-ACs} being not applicable, no identified companion for $S$ is shared by any other internal star.
\end{itemize}
We conclude from the finiteness of the input graph that every satellite of the internal $3^+$-star $S$ receiving one token leads, possibly through an AC, to either
1) a companion $S'$ with $\tau(S') = 0$ or
2) a critical edge into a $2$-star $S'$ with $\tau(S') = 0$ or
3) a heavy $2$-star $S'$ with at least two companions or
4) a critical edge into a $3^+$-star $S'$.
There is at most one critical edge into any internal star, and a $2$-star can be a companion of at most one heavy $2$-star or a $3^+$-star.

Let $\alpha$ denote the total number of satellites of the internal $3^+$-stars each receiving one token,
$\beta_0$ denote the total number of internal $2$-stars each being a companion for some $3^+$-star,
$\beta_1$ denote the total number of internal heavy $2$-stars each having an incoming critical edge,
$\beta_2$ denote the total number of internal $2$-stars each being a companion for one of these $\beta_1$ heavy $2$-stars, and
$\beta_3$ denote the total number of internal $3^+$-stars.
Let $\mcS$ denote the collection of all the above internal stars,
and thus $\mcP \setminus \mcS$ contains the rest of the heavy $2$-stars each associated with at least one companion, and the rest of the non-heavy $2$-stars.

We have $\alpha \le \beta_0 + \beta_1 + \beta_3$, $2 \beta_1 \le \beta_2$, and the total tokens received by $\mcS$ is
\[
\tau(\mcS) = \alpha + 2 \beta_1 \le \beta_0 + \beta_1 + \beta_3 + \beta_2 = |\mcS|.
\]
We also have $\tau(\mcP \setminus \mcS) \le |\mcP \setminus \mcS|$.
Therefore we have proved that $\tau(\mcP) \le |\mcP|$, %, i.e., on average each internal star $S$ has $\tau(S) \le 1$.
and finally since every internal star contains at least three vertices,
\[
\frac {|V(\mcQ)|}{|V(\mcP)|} \le 1 + \frac {\tau(\mcP)}{|V(\mcP)|} \le 1 + \frac {|\mcP|}{|V(\mcP)|} \le 1 + \frac 13 = \frac 43.
\]
This proves the approximation ratio.

For the running time, one sees the most expensive operation is {\sc Trade-along-ACs} of which the inside operation {\sc Trade-$3$} is executed, and it takes $O(|V|^{14.5})$ time.
Since there can be $O(|V|)$ iterations, the overall running time is $O(|V|^{15.5})$.
\end{proof}

\section{Approximating the {\sc $[2, \ell]$-Star Packing} problem when $\ell > 2$}
%==================================================================================================
In the {\sc $[k, \ell]$-Star Packing} problem, where $2 \le k < \ell$ and both $k$ and $\ell$ are fixed integers,
$\ell - k + 1$ $(\ge 2)$ distinct candidate stars are allowed, thereby reduces to the {\em weighted} $(\ell + 1)$-set packing problem.
It thus admits a $(1 + \frac {\ell}2)$-approximation algorithm for all $\ell > 2$ (and a slightly better ratio of $1 + 0.4896 \cdot \ell$ when $\ell \ge 13$)~\cite{TW23}.
In this and the next sections, we propose direct and improved approximation algorithms.
We consider $k = 2$ in this section, and present a $(1 + \frac {\ell}{\ell+1})$-approximation algorithm for all $\ell > 2$.

\subsection{The local search algorithm}
%--------------------------------------------------------------------------------------------------
Similarly as in the last section, we first define the local improvement operations to be repeatedly executed, followed by the performance analysis.
We re-use the names of the operations in the last section since they are very similar, respectively,
except that now the number of satellites in candidate stars is upper unbounded by the constant $\ell$.

Recall that given a $[k, \ell]$-star packing $\mcP$ in the input graph $G = (V, E)$,
i.e., $\mcP$ is a collection of vertex-disjoint stars (called {\em internal}) each with a number of satellites in the given interval $[k, \ell]$,
$V(\mcP)$ denotes the set of vertices in (or, {\em covered} by) the stars of $\mcP$ and the induced subgraph $G[V \setminus V(\mcP)]$ is the {\em remainder} graph with respect to $\mcP$.

\begin{definition}[Operation {\sc Collect}]
\label{def11}
Given a vertex $v$ in the remainder graph, if its degree in the remainder graph is $d \ge k$,
then the operation extracts an $i$-star centered at $v$, where $i = \min\{d, \ell\}$, from the remainder graph and adds it to the current solution $\mcP$.
({\em Comment}: Arbitrary $i$ of the $d$ neighbors of $v$ in the remainder graph are selected as the satellites of the extracted $i$-star;
if $d < k$, then the operation is not applicable on the vertex $v$.)
\end{definition}

The algorithm, denoted as {\sc Approx$[2, \ell]$}, starts with the empty solution,
and in the first a few iterations {\sc Collect} (using $k = 2$) is repeatedly applied on vertices of degree at least $2$ in the remainder graph until impossible.
Note that the degrees of all the vertices can be preprocessed in $O(|E|)$ time;
looking for a vertex and applying the {\sc Collect} operation takes $O(|V|)$ time, followed by updating the degrees of the uncovered vertices in $O(|V|)$ time
(since at most $\ell$ vertices are affected).
That is, each operation {\sc Collect} and an iteration executes the operation takes $O(|V|)$ time.
When no {\sc Collect} operation is applicable, the maximum degree of the remainder graph is at most $1$,
which is stated in the following Observation~\ref{obs05} analogous to Observation~\ref{obs01}.

\begin{observation}
\label{obs05}
When no {\sc Collect} operation is applicable,
\begin{enumerate}
\parskip=0pt
\item the maximum degree of the remainder graph is at most one;
\item the center of an internal $(\ell-1)^-$-star is not adjacent to any uncovered vertex.
\end{enumerate}
\end{observation}

One sees that when no {\sc Collect} operation is applicable, each connected component in the remainder graph is either a vertex or a $1$-star.
As we define more operations below, we remark that after applying one of them in an iteration,
if there is an uncovered vertex adjacent to the center of an internal $(\ell-1)^-$-star, then it is added to the star to become covered,
or if a {\sc Collect} operation becomes applicable, then it is applied to extract and add a feasible star to the current solution.
Such post-processing is executed before an iteration ends, ensuring Observation~\ref{obs05} holds always.

We define operation {\sc Trade-$i$} exactly the same as in \Cref{def02}, except that this time the feasible stars have at least two and at most $\ell$ satellites.
Nevertheless, the algorithm {\sc Approx$[2, \ell]$} uses operation {\sc Trade-$i$} for $i = 1, 2$ only.
In the context of {\sc $[2, \ell]$-Star Packing}, we make new observations below.

\begin{observation}
\label{obs06}
When operation {\sc Trade-$1$} is not applicable,
\begin{enumerate}
\parskip=0pt
\item a satellite of any internal star is adjacent to at most one uncovered vertex;
\item if an internal $2$-star has a satellite adjacent to an outside $1$-star,
	then the other satellite is not adjacent to any uncovered vertex other than those two on the outside $1$-star;
\item at most one adjacency exists among a satellite of an internal $3^+$-star and two uncovered vertices;
\item at most one adjacency exists among two satellites of an internal $4^+$-star and an uncovered vertex;
\item if the center of an internal $3$-star is adjacent to an uncovered vertex,
	then at most one adjacency exists among two satellites of this $3$-star and another uncovered vertex.
\end{enumerate} 
\end{observation}
\begin{proof}
We provide a proof to the first statement.
One will see that it is simple enough and perhaps can be skipped.
The proofs of the other items and the next a few observations are skipped due to their similar levels of simplicity.

For the first item, assume to the contrary that a satellite $v$ of an internal star $S$ is adjacent to two uncovered vertices $x_1$ and $x_2$.
If $S$ is a $3^+$-star, then $v$ can be cut off from $S$ and form with $x_1$ and $x_2$ into a $2$-star;
this is an operation {\sc Trade-$1$} to trade $S$ with two stars.
If $S$ is a $2$-star, then operation {\sc Trade-$1$} is applicable to trade $S$ with a $3^+$-star centered at $v$.
Both result in a contradiction.
\end{proof}

\begin{observation}
\label{obs07}
When operation {\sc Trade-$2$} is not applicable,
\begin{enumerate}
\parskip=0pt
\item at most one adjacency exists among an uncovered vertex and two satellites in two distinct internal $3^+$-stars;
\item if two adjacencies exist among an uncovered vertex, a satellite of a $3^+$-star, and a satellite of a $2$-star $S$,
	then the other satellite of $S$ is not adjacent to any other uncovered vertex;
\item if an uncovered vertex is adjacent to the center of an internal $\ell$-star $S$,
	then no satellite of $S$ is adjacent to the center of any internal $(\ell - 1)^-$-star.
\end{enumerate} 
\end{observation}

We also define an {\em augmenting configuration} (AC) in \cref{def13} below.
We emphasize two major differences from \Cref{def07} that,
the AC here is to ``bring'' an uncovered vertex from a ``distant'' place over while the AC in \cref{def07} is to ``move'' a satellite from a ``distant'' star over,
and the AC here consists of a sequence of ``largest'' $\ell$-stars while the AC in \cref{def07} consists of a sequence of ``smallest'' $k$-stars.
\begin{definition}[Triplet]
\label{def12}
A {\em triplet} $(u, S, v)$ consists of an uncovered vertex $u$, an internal $(\ell - 1)^-$-star $S$, and a satellite $v$ of an internal $\ell$-star.
\end{definition}

\begin{definition}[Augmenting configuration, or AC]
\label{def13}
Given a triplet $(u, S, v)$, if there is a sequence of $j \ge 1$ internal $\ell$-stars $S_1$, $S_2$, $\ldots$, $S_j$ such that
$u$ is adjacent to the center of $S_1$, a satellite of $S_i$ is adjacent to the center of $S_{i+1}$ for every $i = 1, 2, \ldots, j-1$, and $v$ is a satellite of $S_j$,
then we say this sequence of stars $\langle S_1, S_2, \ldots, S_j, S \rangle$ form into an {\em augmenting configuration} abbreviated as {\em AC} for the triplet $(u, S, v)$,
see for an illustration in \Cref{fig05}.
\end{definition}

\begin{figure}[ht]
\centering\scalebox{1.0}{
  \setlength{\unitlength}{1bp}%
  \begin{picture}(314.46, 66.38)(0,0)
  \put(0,0){\includegraphics{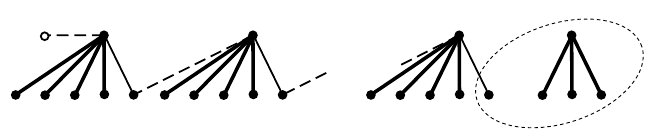}}
  \put(51.40,51.82){\fontsize{11.38}{13.66}\selectfont $c_1$}
  \put(164.57,33.92){\fontsize{11.38}{13.66}\selectfont $\ldots$}
  \put(122.90,51.82){\fontsize{11.38}{13.66}\selectfont $c_2$}
  \put(221.91,51.82){\fontsize{11.38}{13.66}\selectfont $c_j$}
  \put(275.88,51.82){\fontsize{11.38}{13.66}\selectfont $c$}
  \put(11.39,39.91){\fontsize{11.38}{13.66}\selectfont $u$}
  \put(27.59,8.96){\fontsize{11.38}{13.66}\selectfont $S_1$}
  \put(99.01,8.96){\fontsize{11.38}{13.66}\selectfont $S_2$}
  \put(198.98,8.96){\fontsize{11.38}{13.66}\selectfont $S_j$}
  \put(270.44,8.96){\fontsize{11.38}{13.66}\selectfont $S$}
  \put(230.66,8.96){\fontsize{11.38}{13.66}\selectfont $v$}
  \end{picture}%
}
\captionsetup{width=.95\textwidth}
\caption{An augmenting configuration (for $\ell = 5$) for a triplet $(u, S, v)$, in which there is a sequence of $j \ge 1$ internal $\ell$-stars $S_1, S_2, \ldots, S_j$.
	The solid edges are in these internal $\ell$-stars while the dashed edges are in the input graph but not in these $\ell$-stars.
	Swapping the dashed edges into the stars with the thin solid edges makes $u$ covered and $v$ uncovered,
	and $v$ can provide new opportunities for local improvement operations.\label{fig05}}
\end{figure}
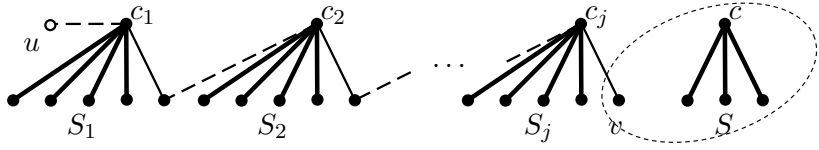

For the current solution $\mcP$, we can implicitly construct a digraph $\mcH$ on the internal $\ell$-stars,
such that if an $\ell$-star $S_1$ has a satellite adjacent to the center of another $\ell$-star $S_2$ then $S_1$ has an arc to $S_2$.
For a triplet $(u, S, v)$, when an AC is found, which can be done by a breadth-first-search traversal in $O(|V|^2)$ time, see for an illustration in \Cref{fig05},
then $u$ and $v$ can be swapped to become covered and uncovered, respectively, thus providing new opportunities for improvement.
This is the operation {\sc Trade-along-AC} defined next.

\begin{definition}[Operation {\sc Trade-along-AC}]
\label{def14}
Given a triplet $(u, S, v)$, let $\mcS$ denote a collection consisting of $S$ and $i-1$ other stars of $\mcP$, for some $1 \le i \le 2$.
If a {\sc Trade-$i$} operation is applicable on $\mcS$ to cover $v$ but not $u$,
and there exists an AC for the triplet $(u, S, v)$ with all the intermediate $\ell$-stars in $\mcP \setminus \mcS$,
then the operation swaps $u$ and $v$ to be covered and uncovered, respectively, and executes the {\sc Trade-$i$} operation on $\mcS$.
\end{definition}

The above operation {\sc Trade-along-AC} is defined to allow for an internal {\sc Trade-$2$} operation, to be used in the next section where $k \ge 3$.
In this section, only {\sc Trade-$1$} is needed, that is, whether {\sc Trade-$1$} is applicable on $S$,
and in fact we only examine for whether $v$ is adjacent to the center of $S$ so that $v$ can be added as a new satellite to $S$ (see for an illustration in \Cref{fig05}).
This ensures an extended Observation~\ref{obs07}(3) where the single $\ell$-star is replaced by the sequence of $\ell$-stars in an AC.

Note that those triplets sharing the same uncovered vertex $u$ can be examined for an AC via a single breadth-first-search traversal in $O(|V|^2)$ time.
It follows that finding and executing a {\sc Trade-along-AC} operation, inside which the triplet $(u, S, v)$ meets the requirement that $v$ is adjacent to the center of $S$,
takes $O(|V|^{3 + 3.5 + 1}) = O(|V|^{7.5})$, an order of time less than the operation in \cref{def08}.
Operations {\sc Trade-$i$} (for $i = 1, 2$) and {\sc Trade-along-AC} are collectively referred to as {\sc Trade}.
Using the {\sc Collect} and {\sc Trade} operations, our algorithm {\sc Approx$[2, \ell]$} repeatedly applies them whenever possible to update the solution.
When none of the operations is applicable, the algorithm returns the achieved solution $\mcP$ as the final solution to the {\sc $[2, \ell]$-Star Packing} problem.

\subsection{Performance analysis}
%--------------------------------------------------------------------------------------------------
Like before we fix an optimal star packing denoted as $\mcQ$ for discussion, the stars in which are referred to as {\em optimal} stars.
We use $\mcP$ to denote the computed star packing, and its stars are called {\em internal} stars;
A vertex is said {\em covered} or {\em uncovered} with respect to $\mcP$.

Using Observation~\ref{obs05}, we categorize the optimal stars into three types exactly the same as in \Cref{def03}, but this time setting $k = 2$.
%For the completeness, we repeat the three types below.
%
%\begin{definition}[Types of optimal stars]
%\label{def15}
%The optimal stars are categorized into three types:
%\begin{itemize}
%\parskip=0pt
%\item In a Type-1 optimal star, the center and one of its satellites are uncovered.
%\item In a Type-2 optimal star, only the center is uncovered.
%\item In a Type-3 optimal star, the center is covered.
%\end{itemize}
%\end{definition}
%
%
We associate a token with each uncovered vertex of $V(\mcQ)$,
and {\em transfer} the tokens of an optimal star in wholes to the covered vertices in the same optimal star according to the following {\em token-transfer rules}.
\begin{definition}[Token-transfer rules]
\label{def15}
\begin{itemize}
\parskip=0pt
\item
	A Type-1 optimal star has two tokens and its all but one satellite are covered.
	We transfer the two tokens to a covered satellite, priority given to one that is the center of an internal star.
\item
	A Type-2 optimal star has one token and its satellites are all covered.
	We transfer the token to a satellite, priority given to one that is the center of an internal star.
\item
	In a Type-3 optimal star, we transfer all the tokens associated with the uncovered satellites, if any, to the center.
\end{itemize}
\end{definition}

We remark a major difference from those rules in \Cref{def09} being that this time some covered vertex has higher priority in receiving the tokens.
In more details, in \cref{def09} the covered satellites in Type-1 and Type-2 optimal stars have the equal chance for receiving tokens,
but in \cref{def15} one that is the center of an internal star has priority.
Another difference is here a Type-3 optimal star can have any number of tokens, unlike in \cref{def03} where the number of tokens is upper unbounded by $k-1 = 1$.

\begin{observation}
\label{obs08}
By the token-transfer rules, in an optimal star, the vertex receiving tokens and all the uncovered vertices induce a connected subgraph inside the optimal star.
\end{observation}

\begin{theorem}
\label{thm03}
The algorithm {\sc Approx$[2, \ell]$} is an $O(|V|^{8.5})$-time $(1 + \frac \ell{\ell+1})$-approximation algorithm for the {\sc $[2, \ell]$-Star Packing} problem, for any $\ell \ge 3$.
\end{theorem}
\begin{proof}
Note that each local search operation increases the number of covered vertices by at least one, and thus the executed total number of operations is in $O(|V|)$.
Among them, {\sc Trade-$2$} takes time in $O(|V|^{6.5})$ while {\sc Trade-along-AC} (with an internal {\sc Trade-$1$}) takes time in $O(|V|^{7.5})$.
We conclude that the overall time complexity of the algorithm is in $O(|V|^{8.5})$.

We next analyze how the vertices receiving non-zero tokens in the token-transfer process are distributed in the internal stars.
Firstly, Observations~\ref{obs05}, \ref{obs06} and~\ref{obs08} tell that the center of an internal $(\ell-1)^-$-star receives no token,
each satellite of a $3^+$-star receives at most $1$ token,
each satellite of an internal $2$-star receives at most $2$ tokens, 
while only the center of an internal $\ell$-star can receive three or more tokens by the third rule.

Consider first an $\ell$-star $S \in \mcP$ of which the center $c$ receives non-zero tokens, of which at least one is from its adjacent uncovered vertex
(see for illustrations in \Cref{fig06}).
Assume a satellite $v$ of $S$ receives the $1$ token from its adjacent uncovered vertex $x$ (cf. Observation~\ref{obs06}(1, 3)) in an optimal star $S^*$.
By the token-transfer rules, $x$ is the only uncovered vertex in $S^*$,
and thus $c \notin V(S^*)$, and furthermore $S^*$ has at least one other covered vertex $u \ne v$ which receives no token.
This vertex $u$ is not in the $\ell$-star $S$ by Observation~\ref{obs06}(4, 5), and therefore $u$ is in another internal star $S' \in \mcP$.
Due to Observation~\ref{obs07}(1, 3), we conclude that $u$ cannot be a satellite if $S'$ is a $3^+$-star or be the center if $S'$ is an $(\ell - 1)^-$-star.
That is, if $u$ is the center of $S'$ then $S'$ is an $\ell$-star and furthermore due to priority $v$ is the center in $S^*$,
or if $u$ is a satellite of $S'$ then $S'$ is a $2$-star and furthermore due to Observation~\ref{obs07}(2) the other satellite of $S'$ receives no token either.
When $S'$ is an $\ell$-star, $(v, u)$ is called the {\em critical edge} from (the center of) $S^*$ into (the center of) $S'$ (illustrated in the left of \Cref{fig06});
when $S'$ is a $2$-star, $S'$ is called a {\em companion} for $S$ identified through the adjacency $v$-$u$ in $S^*$ (illustrated in the right of \Cref{fig06}).
We remark that an $\ell$-star has at most one critical edge into it, and a $2$-star can be a companion at most twice.

\begin{figure}[ht]
\centering\scalebox{1.0}{
  \setlength{\unitlength}{1bp}%
  \begin{picture}(296.15, 82.58)(0,0)
  \put(0,0){\includegraphics{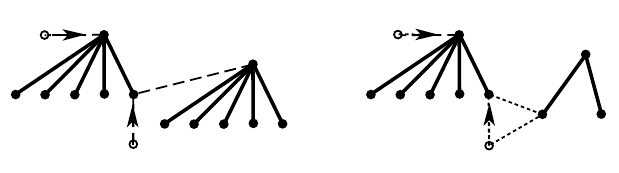}}
  \put(51.40,68.02){\fontsize{11.38}{13.66}\selectfont $c$}
  \put(122.90,53.93){\fontsize{11.38}{13.66}\selectfont $u$}
  \put(221.91,68.02){\fontsize{11.38}{13.66}\selectfont $c$}
  \put(255.00,16.37){\fontsize{11.38}{13.66}\selectfont $u$}
  \put(66.19,9.81){\fontsize{11.38}{13.66}\selectfont $x$}
  \put(237.07,10.08){\fontsize{11.38}{13.66}\selectfont $x$}
  \put(62.90,43.96){\fontsize{11.38}{13.66}\selectfont $v$}
  \put(233.45,42.78){\fontsize{11.38}{13.66}\selectfont $v$}
  \put(21.01,23.39){\fontsize{11.38}{13.66}\selectfont $S$}
  \put(103.23,8.12){\fontsize{11.38}{13.66}\selectfont $S'$}
  \put(193.21,23.16){\fontsize{11.38}{13.66}\selectfont $S$}
  \put(270.26,8.12){\fontsize{11.38}{13.66}\selectfont $S'$}
  \end{picture}%
}
\captionsetup{width=.95\textwidth}
\caption{An illustration of an internal $\ell$-star $S$ of which the center $c$ and a satellite $v$ both receive tokens (indicated by arrows).
	The token of the uncovered vertex $x$ is transferred to $v$ in an optimal star $S^*$, which contains another covered vertex $u$ receiving no token.
	There are only two possibilities for $u$:
	To the left, $u$ is the center of another internal $\ell$-star (here $\ell = 5$); to the right, $u$ is a satellite of an internal $2$-star.\label{fig06}}
\end{figure}
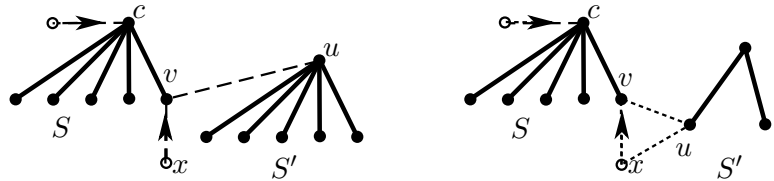

Consider next an $\ell$-star $S \in \mcP$ which has an incoming critical edge from the center of an optimal star into its center $c$.
% is relayed one token from a satellite $v'$ of another $\ell$-star, i.e., $v'$ and $c$ are adjacent.
%We note that $v'$, $c$ and one uncovered vertex are connected inside an optimal star, which has only one uncovered vertex.
If one of its satellites receives no token, then $S$ is said of type-1.
Otherwise, $S$ is said of type-2 and we consider a satellite $v$ of $S$ receiving the $1$ token from its adjacent uncovered vertex $x$ in an optimal star $S^*$
(similarly, see for illustrations in \Cref{fig06}).
By the token-transfer rules, $x$ is the only uncovered vertex in $S^*$ and $c \notin V(S^*)$, and thus $S^*$ contains another covered vertex $u \ne v$.
This vertex $u$ is not in the star $S$ since $S$ is type-2, and therefore $u$ is in another internal star $S' \in \mcP$.
We conclude that due to Observation~\ref{obs07}(1) $u$ cannot be a satellite if $S'$ is a $3^+$-star,
and due to operation {\sc Trade-along-AC} $u$ cannot be the center if $S'$ is an $(\ell - 1)^-$-star.
That is, if $u$ is the center of $S'$ then $S'$ is an $\ell$-star and furthermore due to priority $v$ is the center in $S^*$,
or if $u$ is a satellite of $S'$ then $S'$ is a $2$-star
and furthermore due to operation {\sc Trade-along-AC} and Observation~\ref{obs07}(2) the other satellite of $S'$ receives no token either.
We therefore arrive the same conclusions as in the last paragraph:
When $S'$ is an $\ell$-star, $(v, u)$ is called the {\em critical edge} from (the center of) $S^*$ into (the center of) $S'$ (illustrated in the left of \Cref{fig06});
when $S'$ is a $2$-star, $S'$ is called a {\em companion} for $S$ identified through the adjacency $v$-$u$ in $S^*$ (illustrated in the right of \Cref{fig06}).

One sees that the above exploration process never involves any internal $i$-star with $3 \le i \le \ell-1$.
Let $\alpha_0$ denote the total number of internal $\ell$-stars of which the center receives non-zero tokens
and $\beta$ denote the total number of satellites of these $\alpha_0$ internal $\ell$-stars each receiving one token,
$\alpha_1$ denote the total number of type-1 internal $\ell$-stars (each having an incoming critical edge but at least one satellite receiving no token),
$\alpha_2$ denote the total number of type-2 internal $\ell$-stars (each having an incoming critical edge and every satellite receiving one token),
$\gamma$ denote the total number of internal $2$-stars each being a companion (once or twice).
Let $\mcS$ denote the collection of all the above internal stars,
and thus $\mcP \setminus \mcS$ contains the internal $\ell$-stars of which the center receives no token and which has no incoming critical edge,
the internal $2$-stars which is never a companion,
and the other internal $i$-stars with $3 \le i \le \ell-1$.

Since each of the $\beta$ satellites leads to a type-1 internal $\ell$-star or a companion, possibly through a sequence of type-2 internal $\ell$-stars,
we have $\beta \le \alpha_1 + 2 \gamma$.
The total tokens received by $\mcS$ is
\[
\tau(\mcS) \le \ell \cdot \alpha_0 + \beta + (\ell - 1) \alpha_1 + \ell \cdot \alpha_2 \le \ell \cdot (\alpha_0 + \alpha_1 + \alpha_2) + 2 \gamma.
\]
Note that there are in total $\alpha_0 + \alpha_1 + \alpha_2$ internal $\ell$-stars and $\gamma$ internal $2$-stars in $\mcS$.
Therefore, this average analysis states that in $\mcS$, each $\ell$-star is {\em allocated} with at most $\ell$ tokens and each $2$-star is {\em allocated} with at most $2$ tokens.
By Observations~\ref{obs05}, \ref{obs06} and~\ref{obs08}, in $\mcP \setminus \mcS$, each $i$-star receives at most $i$ tokens, for every $i \in [2, \ell]$.

For each internal star $S \in \mcP$, the total tokens received through token-transfer or allocated by the average analysis, denoted as $\tau(S)$,
is regarded as the number of uncovered vertices in $V(\mcQ)$ that are distributed to the star $S$.
It follows that $|V(\mcQ)| \le |V(\mcP)| + \sum_{S \in \mcP} \tau(S)$, and
\[
\frac {|V(\mcQ)|}{|V(\mcP)|} \le 1 + \frac {\sum_{S \in \mcP} \tau(S)}{|V(\mcP)|} \le 1 + \max_{S \in \mcP} \frac {\tau(S)}{|V(S)|}
	\le 1 + \max_{2 \le i \le \ell}\left\{\frac i{i+1}\right\}
	= 1 + \frac \ell{\ell+1}.
\]
This proves the approximation ratio.
\end{proof}

\section{Approximating the {\sc $[k, \ell]$-Star Packing} problem when $3 \le k < \ell$}
%==================================================================================================
Recall that by $k < \ell$, we have at least two candidate stars to be used in the packings.
When $k \ge 3$, the {\sc $[k, \ell]$-Star Packing} problem has never been studied from the approximation algorithm perspective,
though again it reduces to the weighted $(\ell + 1)$-set packing problem and thus admits a $(1 + \frac {\ell}2)$-approximation algorithm~\cite{TW23}.
Below we show that by using the full capacities of operations {\sc Trade-$i$} and {\sc Trade-along-AC},
i.e., allowing {\sc Trade-$i$} for $i = 1, 2, 3$ and allowing {\sc Trade-$i$} for $i = 1, 2$ inside {\sc Trade-along-AC},
the enhanced algorithm {\sc Approx$[2, \ell]$} is a $(1 + \max\{ \frac {k-1}2, \frac {\ell}{\ell + 1} \cdot \frac {k + 1}3 \})$-approximation algorithm,
for all $3 \le k < \ell$.
Yet, due to the differences, this enhanced algorithm for $3 \le k < \ell$ is denoted as {\sc Approx$[k, \ell]$}.

One sees that, now the most expensive operation is {\sc Trade-along-AC} with an internal {\sc Trade-$2$} operation,
with its time complexity increases to $O(|V|^{9.5})$.
The overall running time for {\sc Approx$[k, \ell]$} is $O(|V|^{10.5})$.

We summarize the following, which integrates several previous observations adjusted for $[k, \ell]$-star packings when $3 \le k < \ell$:

\begin{observation}
\label{obs09}
When no {\sc Collect} or {\sc Trade} operation is applicable,
\begin{enumerate}
\parskip=0pt
\item the maximum degree of the remainder graph is at most $k-1$;
\item the center of an internal $(\ell - 1)^-$-star is not adjacent to any uncovered vertex;
\item for every $1 \le j \le k$, any $j$ satellites of an internal $(k + j)^+$-star together with any $k - j + 1$ uncovered vertices do not form a $k$-star.
\end{enumerate}
\end{observation}

\subsection{Performance analysis}
%--------------------------------------------------------------------------------------------------
The performance analysis for {\sc Approx$[2, \ell]$} can also be adjusted for {\sc Approx$[k, \ell]$}, non-trivially, as follows.
Similarly, we fix an optimal star packing denoted as $\mcQ$ for discussion.
The stars in $\mcQ$ are referred to as {\em optimal} stars.
The computed star packing is denoted as $\mcP$, and its stars are called {\em internal} stars.
A {\em covered} or {\em uncovered} vertex is always with respect to $\mcP$.

Using Observation~\ref{obs09}(1), we categorize the optimal stars into four types, a refinement of \Cref{def03} (see for an illustration in \Cref{fig07}):
\begin{definition}[Types of optimal stars]
\label{def16}
The optimal stars are categorized into four types:
\begin{itemize}
\parskip=0pt
\item In a Type-1 optimal star, the center and $k - 1$ satellites are uncovered.
\item In a Type-2 optimal star, the center and no more than $k - 2$ satellites are uncovered.
\item In a Type-3 optimal star, the center is covered, and at most $k - 1$ satellites are uncovered.
\item In a Type-4 optimal star, the center is covered, but at least $k$ satellites are uncovered.
\end{itemize}
\end{definition}

\begin{figure}[ht]
\centering\scalebox{1.0}{
  \setlength{\unitlength}{1bp}%
  \begin{picture}(220.91, 68.33)(0,0)
  \put(0,0){\includegraphics{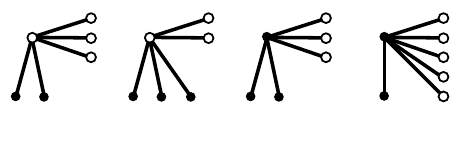}}
  \put(14.66,8.12){\fontsize{11.38}{13.66}\selectfont I}
  \put(71.05,8.12){\fontsize{11.38}{13.66}\selectfont II}
  \put(127.44,8.12){\fontsize{11.38}{13.66}\selectfont III}
  \put(188.53,8.12){\fontsize{11.38}{13.66}\selectfont IV}
  \end{picture}%
}
\captionsetup{width=.95\textwidth}
\caption{Categorizing the optimal stars into four types, using $k = 4$ and $\ell \ge 6$,
	where the filled vertices are covered by the internal stars while blank are uncovered.\label{fig07}}
\end{figure}
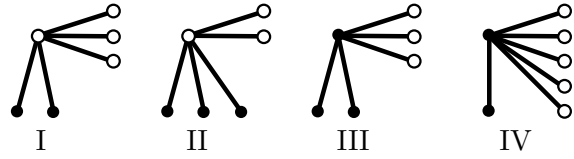

We associate a token with each uncovered vertex of $V(\mcQ)$,
and {\em transfer} the tokens of an optimal star, similarly but this time could be fractionated,
to some covered vertices in the same optimal star according to the following adjusted rules.
We remark again here some covered vertices in Type-2 and Type-3 optimal stars have higher priority in receiving tokens.
\begin{definition}[Token-transfer rules]
\label{def17}
\begin{itemize}
\parskip=0pt
\item
	A Type-1 optimal star has $k$ tokens and all but $k-1$ satellites are covered.
	We transfer these $k$ tokens to a covered satellite.
\item
	A Type-2 optimal star has $x+1$ tokens, where $0 \le x \le k-2$ is the number of uncovered satellites and the other at least $k-x$ satellites are covered.
	If a covered satellite is the center of an internal $\ell$-star, then we transfer all these $x+1$ tokens to it;
	otherwise, we transfer these $x+1$ tokens evenly to $k-x$ covered satellites (each receiving $\frac {x+1}{k-x} \le \frac {k-1}2$ tokens).
\item
	A Type-3 optimal star has $y$ tokens, where $0 \le y \le k-1$ is the number of uncovered satellites and the other at least $k-y$ satellites are covered.
	If the center of the optimal star is the center of an internal $\ell$-star, then we transfer all these $y$ tokens to it;
	otherwise, we transfer these $y$ tokens evenly to the center and $k-y$ covered satellites (each receiving $\frac y{k-y+1} \le \frac {k-1}2$ tokens).
\item
	A Type-4 optimal star has $z$ tokens, where $z \ge k$ is the number of uncovered satellites.
	We transfer these $z$ tokens to the center.
\end{itemize}
\end{definition}

We have the following observation.
\begin{observation}
\label{obs10}
By the token-transfer rules,
\begin{enumerate}
\parskip=0pt
\item
	the covered vertices receiving tokens in an optimal star and those uncovered vertices induce a connected subgraph inside the optimal star;
\item
	the center of an internal $(\ell - 1)^-$-star receives at most $\frac {k-1}2$ tokens, and if it receives non-zero tokens then it is by the third rule;
\item
	the center of an internal $\ell$-star can receive up to $\ell$ tokens;
\item
	a satellite of internal $k$-stars can receive up to $k$ tokens but not by the fourth rule,
	and if it receives $k$ tokens, i.e., by the first rule, then each of the other satellites receives at most $\frac {k-1}2$ tokens;
\item
	a satellite of internal $(k + 1)^+$-stars receives at most $\frac {k-1}2$ tokens.
\end{enumerate}
\end{observation}

Similarly as in \Cref{def05}, if a satellite of an internal $k$-star $S$ receives $k$ tokens then $S$ is called {\em heavy}.
Furthermore, when $S$ is heavy, for another internal $k$-star $S'$,
if a satellite $v$ of $S$ and a satellite $v'$ of $S'$ are co-recipients in an optimal star each receiving $\frac {k-1}2$ tokens,
then $S'$ is a {\em companion} for $S$ (see for an illustration in \Cref{fig01}).

We have the same conclusions on heavy stars and their companions as stated in Lemma~\ref{lemma01}, but their proofs require the new token-transfer rules in \Cref{def17}.
\begin{lemma}
\label{lemma02}
\begin{enumerate}
\parskip=0pt
\item
	The center of a heavy star receives at most $\frac {k-2}3$ tokens;
\item
	the companion star of a heavy star is not heavy;
\item
	the center of a companion star receives at most $\frac {k-2}3$ tokens;
\item
	no two heavy stars share a common companion.	
\end{enumerate}
\end{lemma}
\begin{proof}
We first note that Observation~\ref{obs03} on $k$-stars are the same as Observation~\ref{obs10} on $k$-stars.

Let $S$ be a heavy star, i.e., one satellite receives $k$ tokens in an optimal star $S^*$.
By Observation~\ref{obs10}(4) each of the other satellites can receive at most $\frac {k-1}2$ tokens.
If the center $c$ of $S$ receives $\frac {k-1}2$ tokens, then by Observation~\ref{obs10}(2) the tokens are transferred in the optimal star $S^*_1$ by the third rule,
with its co-recipient $u$ being the center of $S^*_1$ and thus adjacent to the $k-1$ uncovered vertices.
It follows from Observation~\ref{obs09}(2) and the third token-transfer rule that $u$ is a satellite in an internal star $S'$,
but then operation {\sc Trade-$2$} is applicable to trade $S$ and $S'$ with the two stars centered at the center of $S^*$ and $u$,
respectively, and another star centered at the center of $S'$ when it is a $(k+1)^+$-star, a contradiction.
This proves item 1.

Item 2 can be proved identically as in the proof of \Cref{lemma01}, and we will not repeat here.

Next, let $S'$ be a companion for $S$ such that a satellite $v$ of $S$ and a satellite $v'$ of $S'$ are co-recipients in an optimal star $S^*_1$ each receiving $\frac {k-1}2$ tokens.
If the center $c'$ of $S'$ receives $\frac {k-1}2$ tokens,
then the same the tokens are transferred to $c'$ in the optimal star $S^*_2$ by the third rule,
with its co-recipient $u'$ being the center of $S^*_2$ and thus adjacent to the $k-1$ uncovered vertices.
It follows from Observation~\ref{obs09}(2) and the third token-transfer rule that $u'$ is a satellite in an internal star $S''$,
but then operation {\sc Trade-$3$} is applicable to trade $S$, $S'$ and $S''$ with three or four stars constructed from
these three stars and the three involved optimal stars $S^*$, $S^*_1$ and $S^*_2$ to cover at least one more vertex, a contradiction.
This proves item 3.

Lastly, if $S'$ is also a companion for another heavy star $S''$,
and assume a satellite of $S''$ receives $k$ tokens in an optimal star $S^*_1$,
a satellite of $S$ and a satellite of $S'$ are co-recipients in an optimal star $S^*_2$ each receiving $\frac {k-1}2$ tokens,
and a satellite of $S''$ and a satellite of $S'$ are co-recipients in an optimal star $S^*_3$ each receiving $\frac {k-1}2$ tokens,
then operation {\sc Trade-$3$} would be applicable to trade $S$, $S'$ and $S''$ with four stars constructed from
these four involved optimal stars $S^*$, $S^*_1$, $S^*_2$ and $S^*_3$ to cover at least $k+1$ more vertices, a contradiction.
This proves the last item of the lemma.
\end{proof}

In the proof of the following main theorem in this section, we examine the total tokens received by (all the vertices of) each internal star $S$, denoted as $\tau(S)$.

\begin{theorem}
\label{thm04}
The algorithm {\sc Approx$[k, \ell]$} is an $O(|V|^{10.5})$-time $(1 + \max \{ \frac{k-1}2, \frac{(k+1)\ell}{3(\ell+1)} \})$-approximation algorithm
for the {\sc $[k, \ell]$-Star Packing} problem,
where $3 \le k < \ell$.
\end{theorem}
\begin{proof}
The overall running time of $O(|V|^{10.5})$ has been discussed at the beginning of the section.
We distinguish the internal stars.

{\bf Case 1: $i$-stars where $k+1 \le i \le \ell-1$.}
Consider such an internal $i$-star $S$.
By Observation~\ref{obs10}(2, 5), every vertex in $S$ receives at most $\frac {k-1}2$ tokens and therefore $\tau(S) \le \frac {k-1}2 \cdot |V(S)|$.

{\bf Case 2: $k$-stars}.
Let $S$ be an internal $k$-star.
Observation~\ref{obs10}(2) says the center of $S$ receives at most $\frac {k-1}2$ tokens;
Observation~\ref{obs10}(4) says if $S$ is heavy then each of the other satellites receives at most $\frac {k-1}2$ tokens,
and by \Cref{lemma02} the center receives at most $\frac {k-2}3$ tokens.
Therefore, when $S$ is non-heavy, $\tau(S) \le \frac {k-1}2 \cdot |V(S)|$;
when $S$ is heavy but no other satellite receives $\frac {k-1}2$ tokens, $\tau(S) \le k + k \cdot \frac {k-2}3 = \frac k3 \cdot |V(S)| \le \frac {k-1}2 \cdot |V(S)|$ (due to $k \ge 3$).

When $S$ is heavy and a satellite $v$ of $S$ receives $\frac {k-1}2$ tokens in the optimal star $S^*$,
then $v$ receives tokens by the second rule, since otherwise {\sc Trade-2} would be applicable on $S$ and the internal star $S'$ containing the co-recipient $v'$.
Furthermore, it follows from Observation~\ref{obs09}(2) and the second rule that $v'$ is a satellite in $S'$, and $S' \ne S$ is a $k$-star.
That is, we have identified a companion for $S$ through the path $v$-$c^*$-$v'$ (where $c^*$ is the center of $S^*$, see for an illustration in \cref{fig01}).
Assume $j$ satellites of $S$ each receives $\frac {k-1}2$ tokens, then we can identify $j$ unique companions $S'_1, S'_2, \ldots, S'_j$ for $S$.
Let $\mcS = \{S, S'_1, S'_2, \ldots, S'_j\}$.
By \Cref{lemma02} we have 
\[
\tau(\mcS) \le \left(k + (k-j) \cdot \frac {k-2}3 + j \cdot \frac {k-1}2\right) + j \cdot \left(\frac {k-2}3 + k \cdot \frac {k-1}2\right) \le \frac {k-1}2 \cdot |V(\mcS)|.
\]

{\bf Case 3: $\ell$-stars}.
Let $S$ be an internal $\ell$-star.
Observation~\ref{obs10}(5) says each satellite receives at most $\frac {k-1}2$ tokens.
Observation~\ref{obs10}(3) says the center of $S$ can receive up to $\ell$ tokens, and in more details,
it receives $k$, $x+1$ $(\le k-1)$, $y$ $(\le k-1)$, and $z$ $(k \le z \le \ell)$ tokens by the first, second, third and fourth token-transfer rules, respectively.
Furthermore, if the center of $S$ receives $y$ tokens by the third token-transfer rule,
then it is also the center of the optimal star;
otherwise, it receives only $\frac y{k-y+1} \le \frac {k-1}2$ tokens.

When the center of $S$ receives no more than $\frac {k-1}2$ tokens, the total token received by $S$ is
$\tau(S) \le (\ell + 1) \cdot \frac {k-1}2 = \frac {k-1}2 \cdot |V(S)|$.
When every satellite of $S$ receives no more than $\frac {k-2}3$ tokens, the total token received by $S$ is
$\tau(S) \le \ell + \ell \cdot \frac {k-2}3 = (\ell+1) \cdot \frac {\ell \cdot (k+1)}{3 (\ell+1)} = \frac {(k+1) \ell}{3 (\ell+1)} \cdot |V(S)|$.

Assume next the center $c$ of $S$ receives strictly more than $\frac {k-1}2$ tokens and
a satellite $v$ receives $\frac {k-1}2$ tokens in an optimal star $S^*$ and the co-recipient is $v'$.
We thus see that $c$ is adjacent to an uncovered vertex $u$ and receives its token,
and that the vertex $v'$ is not in $S$ since otherwise {\sc Trade-1} would be applicable on $S$, and therefore $v'$ is in another internal star $S' \in \mcP$.
Due to operation {\sc Trade-2}, we conclude that $v'$ cannot be a satellite if $S'$ is a $(k + 1)^+$-star or be the center if $S'$ is an $(\ell - 1)^-$-star.
That is, if $v'$ is the center of $S'$ then $S'$ is an $\ell$-star or if $v'$ is a satellite of $S'$ then $S'$ is a $k$-star.

In the former case, we have identified a sequence of $\ell$-stars that can be used to swap $u$ and $v$ to become covered and uncovered, respectively,
so that replacing operation {\sc Trade-$i$} by operation {\sc Trade-along-AC} we can repeat the argument on any satellite of $S'$ if it receives $\frac {k-1}2$ tokens.
Or otherwise, $S'$ is called a {\em companion} of $S$, that every one of its satellites receives at most $\frac {k-2}3$ tokens and it is not a companion for any other $\ell$-star.
Due to the finiteness of the input graph, we always reach the latter case, where the $k$-star $S'$ is identified through one of its satellite $v'$ receiving $\frac {k-1}2$ tokens,
and it is also called a {\em companion} of $S$.
Similarly as the proof of \Cref{lemma02}, replacing operation {\sc Trade-$i$} by operation {\sc Trade-along-AC},
the companion $S'$ is not heavy, its center can receive at most $\frac {k-2}3$ tokens, and it is not a companion for any other $\ell$-star.

We thus have proved that for the satellite $v$ of $S$ we can identify a unique companion star $S'$,
which can take over $\frac {k-1}2 - \frac {k-2}3$ tokens of $v$ to reach a total of at most $\frac {k-1}2 \cdot |V(S')|$.
That is, every satellite of $S$ retains $\frac {k-2}3$ tokens, and thus the retained total token is at most $\frac {(k+1) \ell}{3 (\ell+1)} \cdot |V(S)|$.
Therefore, on average, each internal star $S$ has
\[
\tau(S) \le \max\left\{ \frac {k-1}2 \cdot |V(S)|, \frac {(k+1) \ell}{3 (\ell+1)} \cdot |V(S)|\right\} = \max\left\{ \frac {k-1}2, \frac {(k+1) \ell}{3 (\ell+1)}\right\} \cdot |V(S)|.
\]
It follows from $|V(\mcQ)| \le |V(\mcP)| + \sum_{S \in \mcP} \tau(S)$ that
\[
\frac {|V(\mcQ)|}{|V(\mcP)|} \le 1 + \frac {\sum_{S \in \mcP} \tau(S)}{|V(\mcP)|} \le 1 + \max_{S \in \mcP} \frac {\tau(S)}{|V(S)|}
	= 1 + \max\left\{\frac{k-1}2, \frac{(k+1)\ell}{3(\ell + 1)} \right\}.
\]
This proves the approximation ratio.
\end{proof}

\section{An approximation lower bound for {\sc $[2, \infty)$-Star Packing}}
%==================================================================================================
In this section, we prove the APX-hardness for the {\sc $[2, \infty)$-Star Packing} problem via a gap-preserved reduction from an APX-hard variant of the {\sc Max $3$CNF-SAT} problem.

\begin{definition}[{\sc Max $3$CNF-SAT}]
\label{def18}
Given a set $V$ of $n$ boolean variables and a conjunctive normal form (CNF) formula of $m$ clauses each containing exactly three literals,
the goal of the {\sc Max $3$CNF-SAT} problem is to find a truth assignment for the variables to satisfy the maximum number of clauses.
When every variable appears in at most $B$ clauses, the problem is denoted as {\sc Max $3$CNF-SAT-$B$}.
\end{definition}

It is NP-hard to approximate the {\sc Max $3$CNF-SAT} problem within $\frac 87$~\cite{Has97,Has01},
even assuming there exists a truth assignment satisfying all the $m$ clauses in the $3$CNF formula.
For {\sc Max $3$CNF-SAT-$3$}, differently, it is NP-hard to approximate within $\frac {62}{61}$, but more specifically,
for any $0 < \delta < \frac 1{124}$ there are instances for which it is NP-hard to distinguish whether there is a truth assignment satisfying at least $(1 - \delta) m$ clauses
or every truth assignment satisfies at most $(\frac {61}{62} + \delta) m$ clauses~\cite{EK06}.

Let us fix a {\sc Max $3$CNF-SAT-$3$} instance for discussion, in which the set of variables is $V = \{x_1, x_2, \ldots, x_n\}$ and the set of clauses is $C = \{c_1, c_2, \ldots, c_m\}$.
In the clause $c_j = x_{j,1} \lor x_{j,2} \lor x_{j,3}$, each literal $x_{j,\ell}$ is the positive or negative form of a variable;
and every variable $x_i$ occurs $t_i \in \{2,3\}$ times across all the clauses, once or twice in positive but exactly once in negative.
We extend the reduction by Xi et al.~\cite{XLL24} for proving the NP-hardness of {\sc $[2, \infty)$-Star Packing},
to construct a subcubic graph $G$ (see for an illustration in \Cref{fig08} and \Cref{fig09}).
\begin{itemize}
\parskip=0pt
\item ({\em Clause gadget})
	For each clause $c_j = x_{j,1} \lor x_{j,2} \lor x_{j,3}$,
	create the clause gadget $U_j$ consisting of $8$ vertices and $7$ internal edges as illustrated in \Cref{fig08},
	among which three vertices are labeled by the literals $x_{j,1}, x_{j,2}, x_{j,3}$, respectively
	(the other five are called {\em private} vertices with respect to the gadget).
\item ({\em Variable gadget})
	For each variable $x_i$,
	\begin{itemize}
	\parskip=0pt
	\item
		if $t_i = 2$, create the variable gadget $V_i$ consisting of $15$ vertices and $17$ internal edges as illustrated in \Cref{fig09}(a),
		among which two vertices are labeled by $x_i$ and $\overline{x}_i$, respectively
		(the other thirteen are called private vertices);
	\item
		if $t_i = 3$, create the variable gadget $V_i$ consisting of $17$ vertices and $18$ internal edges as illustrated in \Cref{fig09}(b),
		among which three vertices are labeled by $x_i$, $x_i$ and $\overline{x}_i$, respectively
		(the other fourteen are called private vertices).
	\end{itemize}
\item
	For each labeled vertex in a clause gadget, connect it to a distinct vertex labeled by the same literal in the variable gadget.
	Note that for every literal, there are an identical number of vertices in the clause gadgets and in the variable gadget labeled by it, 	which is either one or two.
\end{itemize}
One sees that 1) $\frac 23 n < m \le n$, 2) there are at most $8 m + 17 n < 33.5 m$ vertices and at most $7 m + 18 n + 3 m \le 37 m$ edges in the constructed graph $G$,
and 3) $G$ is subcubic.

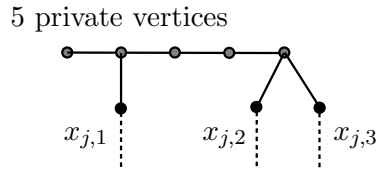
\begin{figure}[ht]
\centering\scalebox{1.0}{
  \setlength{\unitlength}{1bp}%
  \begin{picture}(177.70, 73.10)(0,0)
  \put(0,0){\includegraphics{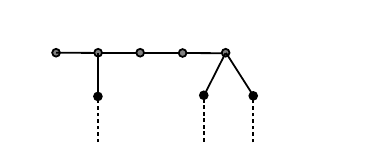}}
  \put(25.61,16.07){\fontsize{11.38}{13.66}\selectfont $x_{j,1}$}
  \put(5.67,58.54){\fontsize{11.38}{13.66}\selectfont $5$ private vertices}
  \put(77.58,16.07){\fontsize{11.38}{13.66}\selectfont $x_{j,2}$}
  \put(126.81,16.07){\fontsize{11.38}{13.66}\selectfont $x_{j,3}$}
  \end{picture}%
}
\captionsetup{width=.95\textwidth}
\caption{The gadget $U_j$ for the clause $c_j = x_{j,1} \lor x_{j,2} \lor x_{j,3}$,
	containing $5$ private vertices and $3$ vertices labeled by the three literals respectively.\label{fig08}}
\end{figure}

\begin{figure}[ht]
\centering\scalebox{1.0}{
  \setlength{\unitlength}{1bp}%
  \begin{picture}(310.73, 98.03)(0,0)
  \put(0,0){\includegraphics{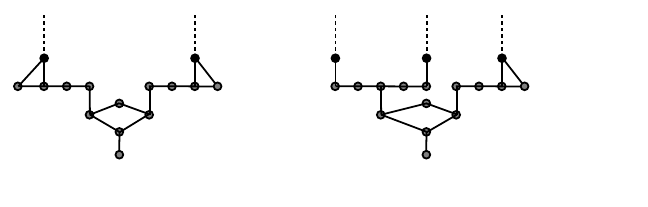}}
  \put(209.97,27.62){\fontsize{11.38}{13.66}\selectfont $14$ private vertices}
  \put(145.47,71.03){\fontsize{11.38}{13.66}\selectfont $x_i$}
  \put(189.24,71.03){\fontsize{11.38}{13.66}\selectfont $x_i$}
  \put(244.59,71.03){\fontsize{11.38}{13.66}\selectfont $\bar{x}_i$}
  \put(51.94,8.12){\fontsize{11.38}{13.66}\selectfont (a)}
  \put(200.18,8.12){\fontsize{11.38}{13.66}\selectfont (b)}
  \put(62.61,27.62){\fontsize{11.38}{13.66}\selectfont $13$ private vertices}
  \put(5.67,71.03){\fontsize{11.38}{13.66}\selectfont $x_i$}
  \put(97.23,71.03){\fontsize{11.38}{13.66}\selectfont $\bar{x}_i$}
  \end{picture}%
}
\captionsetup{width=.95\textwidth}
\caption{The gadget $V_i$ for the variable $x_i$:
	(a) contains $13$ private vertices and $2$ vertices labeled by the two literals when $x_i$ appears two times in the clauses;
	(b) contains $14$ private vertices and $3$ vertices labeled by the three literals when $x_i$ appears three times in the clauses.\label{fig09}}
\end{figure}
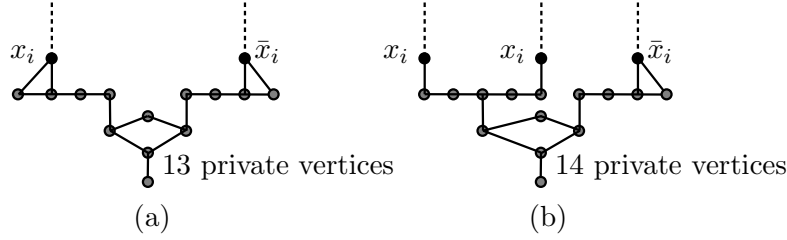

Note that while the decision version of {\sc $[2, \infty)$-Star Packing} asks for a packing (called {\em perfect}) to cover all the vertices,
the maximization problem seeks for a packing to cover the maximum number of vertices.
We next prove some structural properties for truth assignments and feasible $[2, \infty)$-star packings.

\begin{lemma}
\label{lemma03}
Given a truth assignment for the instance of {\sc Max $3$CNF-SAT-$3$} that leaves some clauses unsatisfied,
another truth assignment, if needed, can be obtained in $O(n)$ time to satisfy at least the same number of clauses while the variables in the unsatisfied clauses are distinct,
i.e., no variable appears in two unsatisfied clauses.
\end{lemma}
\begin{proof}
The proof is done by construction to flip the truth value of the variable that appears in two unsatisfied clauses, if any.
\end{proof}

A truth assignment as stated in Lemma~\ref{lemma03} is said {\em canonical}.

Similarly as observed by Xi et al.~\cite{XLL24}, given a $[2, \infty)$-star packing $\mcP$ in $G$,
if one of the labeled vertex in a clause gadget $U_j$ is the center of a $2$-star, then at least one private vertex will not be covered.
One can swap to use private vertices as centers to cover at least the same number of vertices.
Next, another similar observation holds that the two stars in $\mcP$ centered at the private vertices of $U_j$ need at least one of the three labeled vertices to be their satellite,
or otherwise at least two private vertices would be uncovered.
Consider the gadget $V_i$ for a variable $x_i$ which appears only twice in the clauses, say $c_j$ and $c_\ell$.
If for both $U_j$ and $U_\ell$, the stars in $\mcP$ need $x_i$ and $\bar{x}_i$ to be their satellites,
then exactly one vertex in $V_i$ is not covered;
additionally, if for both $U_j$ and $U_\ell$, the stars in $\mcP$ need none of $x_i$ and $\bar{x}_i$ to be their satellites,
then arbitrarily one of them is made a satellite.
Analogously, for the gadget $V_i$ for a variable $x_i$ appearing three times (two positive and one negative) in the clauses,
if $\bar{x}_i$ and at least one $x_i$ are needed as satellites of stars centered at the private vertices in the corresponding clause gadgets,
then exactly one vertex in $V_i$ is not covered;
additionally, if $\bar{x}_i$ is not needed, then both $x_i$'s are made satellites,
or if $\bar{x}_i$ is needed but the two $x_i$'s are not, then only $\bar{x}_i$ is made a satellite.
The packing $\mcP$ is said {\em canonical} if all the uncovered vertices are in the variable gadgets and
for these variable gadgets all the corresponding labeled vertices in the clause gadgets are needed as satellites.

\begin{lemma}
\label{lemma04}
A canonical truth assignment for the {\sc Max $3$CNF-SAT-$3$} instance leaving $\ell$ clauses unsatisfied one-to-one corresponds
to a canonical $[2, \infty)$-star packing in the graph $G$ leaving $\ell$ vertices uncovered.
\end{lemma}
\begin{proof}
Given a canonical truth assignment, for each clause $c_j = x_{j,1} \lor x_{j,2} \lor x_{j,3}$,
if a literal $x_{j,z}$ is true then the vertex in $U_j$ labeled by the literal is a satellite of a star center at the private vertex in $U_j$.
This gives rise to an initial $[2, \infty)$-star packing $\mcP_0$ in which all the vertices in satisfied clause gadgets and vertex gadgets covered,
but exactly $2 \ell$ private vertices, two in each unsatisfied clause gadget, uncovered.
Next, for each unsatisfied clause $c_j$, take an arbitrary literal and swap to cover the vertex labeled by the literal inside $U_j$;
by Lemma~\ref{lemma03} this process will cover all the vertices in $U_j$ but leave exactly one private vertex in the corresponding variable gadget uncovered.
Repeating this process, we achieve a packing $\mcP$ which leaves exactly $\ell$ vertices uncovered.

Conversely, given a canonical $[2, \infty)$-star packing with $\ell$ vertices uncovered, these $\ell$ vertices are in $\ell$ distinct variable gadgets;
and for each such gadget $V_i$, the corresponding vertices labeled by $x_i$ and $\bar{x}_i$ in the clause gadgets are needed as satellites.
We assign the variable $x_i$ true.
For each variable gadget $V_i$ containing no uncovered vertex, if the corresponding vertices labeled by $x_i$ in the clause gadgets are covered as satellites, then set $x_i$ true,
or otherwise set $x_i$ false.
This truth assignment makes exactly $\ell$ clauses unsatisfied, and thus proves the lemma.
\end{proof}

\begin{theorem}
\label{thm05}
It is NP-hard to approximate the {\sc $[2, \infty)$-Star Packing} problem within $\frac {2077}{2076} - \epsilon$, for any $\epsilon > 0$.
\end{theorem}
\begin{proof}
Recall that it is NP-hard to distinguish between {\sc Max $3$CNF-SAT-$3$} instances with $m$ clauses
that there is a truth assignment satisfying at least $(1 - \delta) m$ clauses or every truth assignment satisfies at most $(\frac {61}{62} + \delta) m$ clauses,
for any $0 < \delta < \frac 1{124}$~\cite{EK06}.
By \cref{lemma04}, this gap of $\frac 1{62} m$ is maintained on graphs of order $33.5 m$,
and thus it is NP-hard to approximate within $33.5 / (33.5 - \frac 1{62}) = \frac {2077}{2076} \approx 1.000481$.
This completes the proof.
\end{proof}

\section{Conclusions}
%==================================================================================================
In this paper, we studied four NP-hard sequential variants of the star packing problem from the approximation algorithm perspective,
where the number of satellites in a candidate star is in a given interval $[k, \ell]$, where $2 \le k < \ell$ and $\ell$ is a fixed integer or infinity.
We contributed four improved or the first approximation algorithms for them, respectively.
Namely, for {\sc $[k, \infty)$-Star Packing} where $k \ge 3$, we presented a $\frac {k+1}2$-approximation algorithm,
improving the previous best $\frac {(k+1)^2}{2k+1}$-approximation algorithm by Hu et al.~\cite{HZC25,HZC26};
for {\sc $[2, \infty)$-Star Packing}, we designed a $\frac 43$-approximation algorithm, improving the previous best ratio of $\frac 32$.
The {\sc $[k, \ell]$-Star Packing} problem, where $2 \le k < \ell$, is first studied by us,
and when $k = 2$, we proposed the first $(1 + \frac \ell{\ell+1})$-approximation algorithm;
when $k \ge 3$, a slight extension leads to a $(1 + \max\left\{\frac {k-1}2, \frac {(k+1) \ell}{3 (\ell+1)}\right\})$-approximation algorithm.

The main algorithm design and analysis techniques are local search coupled with amortized analysis;
additionally, inspired by various augmenting structures observed in the past work of Edmonds~\cite{Edm65b} and Hell and Kirkpatrick~\cite{HK84,HK86},
we explored augmenting configurations that help bridge two distinct local structures for applying a local improvement operation.
It would be interesting to see other algorithm design techniques than local search to design further better approximation algorithms,
or other global packing structures that can be constructed in polynomial time to ensure the packing quality.

On the inapproximability, we proved a lower bound of $\frac {2077}{2076}$ for the last remaining case of {\sc $[2, \infty)$-Star Packing}.
For {\sc $[k, \infty)$-Star Packing} where $k \ge 3$, we can achieve better lower bounds than the previous one of $\frac {12k + 19}{12k + 18.994}$,
but all these bounds are far away from the state-of-the-art approximation ratios.
It thus requires continuous efforts in narrowing down the gaps for each of these sequential star packing problems.

%==================================================================================================
%%
%% Bibliography
%%
%% Please use bibtex, 

%\bibliography{../BiBTeX/setcover,../BiBTeX/mypapers,../BiBTeX/general} %,references

%==================================================================================================
\end{document}